\RequirePackage{fix-cm}
\documentclass[smallextended]{svjour3}       
\smartqed  
\usepackage{graphicx}
\usepackage{amssymb}
\usepackage{amsmath}
\usepackage{bbm}
\usepackage[sort&compress, square, numbers]{natbib}
\newcommand{\E}{\mathbb{E}}
\begin{document}

\title{
A New Poisson Noise Filter based on
Weights Optimization
}


\author{Qiyu JIN         \and
        Ion Grama        \and
        Quansheng Liu
}


\institute{Qiyu JIN \at
              Universit\'{e} de Bretagne-Sud, Campus de Tohaninic, BP 573,
56017 Vannes, France \\
Universit\'{e} Europ\'{e}enne de Bretagne, France\\
              \email{qiyu.jin@univ-ubs.fr}           
           \and
           Ion Grama \at
              Universit\'{e} de Bretagne-Sud, Campus de Tohaninic, BP 573,
56017 Vannes, France\\
Universit\'{e} Europ\'{e}enne de Bretagne, France\\
              \email{ion.grama@univ-ubs.fr}
           \and
           Quansheng Liu \at
              Universit\'{e} de Bretagne-Sud, Campus de Tohaninic, BP 573,
56017 Vannes, France\\
Universit\'{e} Europ\'{e}enne de Bretagne, France\\
              \email{quansheng.liu@univ-ubs.fr}
}

\maketitle

\begin{abstract}
We propose a new image denoising algorithm when the data is contaminated by a
Poisson noise. As in the Non-Local Means filter, the proposed algorithm is based on a weighted linear combination of the observed image. But in contract to the latter where the  weights are defined by a Gaussian kernel, we propose to choose them in an optimal way. First some "oracle" weights are defined by minimizing a very tight upper bound of the Mean Square Error.
For a practical application the weights are estimated from the observed image.
We prove that the proposed filter converges at the usual optimal rate to the true image.
Simulation results are presented to compare the performance of the presented filter with conventional filtering methods.
\keywords{Poisson noise\and Mean Square Error\and oracle
estimate \and Optimal Weights Filter }
\end{abstract}

\section{Introduction}
In a variety of applications, ranging from nuclear medicine to night vision
and from astronomy to traffic analysis, data are collected by counting a
series of discrete events, such as photons hitting a detector or vehicles
passing a sensor.   Many such
problems can be viewed as the recovery of the intensity from the indirect Poisson data.
The measurements are often inherently noisy due to low
count levels, and we wish to reconstruct salient features of the underlying
phenomenon from these noisy measurements as accurately as possible.

There are many types of methods to reconstruct the image contaminated by the
Poisson noise. The most popular method  is
performed through a Variance Stabilizing Transform (VST) with the  following
three-step procedure. First, the variance of the Poisson distribution  is
stabilized  by applying a VST. So that the
transformed data are approximately homoscedastic and Gaussian. The VST can
be an Anscombe root transformation (Anscombe \citep{ANSCOMBE1948TRANSFORMATION}
 and Borovkov \citep{borovkov2000estimates}), multiscal VSTs (Bardsley and  Luttman
\citep{ZHANG2008WAVELETS}), Conditional Variance Stabilization (CVS) (Jansen
\citep{jansen2006multiscale}), or Haar-Fisz transformation (Fryzlewicz and  Nason
\citep{FRYZLEWICZ2007GOES,FRYZLEWICZ2004HAAR}). Second, the
noise is removed using a conventional denoising algorithm for additive
Gaussian white noise, see for example
Buades, Coll and Morel (2005 \citep{buades2005review}),
Kervrann (2006 \citep{kervrann2006optimal}), Aharon and Elad and Bruckstein
(2006 \citep{aharon2006rm}), Hammond and Simoncelli (2008 %
\citep{hammond2008image}), Polzehl and Spokoiny (2006 %
\citep{polzehl2006propagation}), Hirakawa and Parks (2006 %
\citep{hirakawa2006image}), Mairal, Sapiro and Elad (2008 %
\citep{mairal2008learning}), Portilla, Strela, Wainwright and
Simoncelli (2003 \citep{portilla2003image}), Roth and Black (2009 %
\citep{roth2009fields}), Katkovnik, Foi, Egiazarian, and Astola (2010 %
\citep{Katkovnik2010local}), Dabov, Foi, Katkovnik and Egiazarian (2006 %
\citep{buades2009note}),
 Abraham, Abraham, Desolneux and  Li-Thiao-Te (2007 \citep{Abraham2007significant}), and Jin, Grama and Liu (2011 \citep{JinGramaLiuowf}).
 Third, an inverse transformation is applied to the
denoised signal, obtaining the estimate of the signal of interest. Makitalo
and Foi (2009 \citep{MAKITALO2009INVERSION} and 2011 %
\citep{makitalo2011optimal}) focus on this last step, and introduce the Exact
Unbiased Inverse (EUI) approach. Zhang, Fadili, and Starck (2008 %
\citep{ZHANG2008WAVELETS}), Lefkimmiatis, Maragos, and Papandreou (2009 %
\citep{lefkimmiatis2009bayesian}), Luisier,  Vonesch, Blu and Unser (2010 \citep{luisier2010fast})
 improved both  the stabilization and   the inverse transformation.

Regularization based on a total variation seminorm has also attracted
significant attention, see for example Beck and Teboulle (2009 %
\citep{beck2009fast}), Bardsley and Luttman (2009 \citep{bardsley2009total}%
), Setzer, Steidl and Teuber (2010 \citep{setzer2010deblurring}). Nowak and
Kolaczyk (1998 \citep{nowak1998multiscale} and 2000 \citep{nowak2000statistical})
have investigated reconstruction algorithms specifically designed for the
Poisson noise without the need of VSTs.

In this paper, we introduce a new algorithm to restore the Poisson noise without
using VST's. We combine the special properties
 of the
Poisson distribution and the idea of Optimal Weights Filter %
\citep{JinGramaLiuowf} for removing efficiently the Poisson noise. The use
of the proposed filter is justified both from the theoretical point of view
by  convergence theorems, and by simulations which show that the
filter is very effective.

The paper is organized as follows. Our main results are presented in Section %
\ref{Sec:oracle} where we construct an adaptive estimator and give an
estimation of its rate of convergence. In Section \ref{Sec:simulations}, we
present our simulation results with a brief analysis. Proofs of the main
results are deferred to Section \ref{Sec:Appendix Proofs}.

\section{\label{Sec:oracle} Construction of the estimator and its convergence}
\subsection{The model and the notations}

We
suppose that the original image of the object being photographed is a
integrable two-dimensional function $f(x)$, $x\in (0,1]\times(0,1]$. Let the
mean value of $f$ in a set $\mathbf{B}_x$ be
\begin{equation*}
\Lambda(\mathbf{B}_x)=N^2\int\limits_{\mathbf{B}_x}f(t)dt.  \label{Lam001}
\end{equation*}
Typically we observe a discrete data set of counts $\mathbf{Y}=\{\mathcal{N}(\mathbf{B}%
_x)$\}, where $\mathcal{N}(\mathbf{B}_x)$ is a Poisson random variable of intensity $%
\Lambda(\mathbf{B_x})$. We consider that if
$
\mathbf{B}_x \cap \mathbf{B}_y=\emptyset,
$
 then $\mathcal{N}(\mathbf{B}_x)$ is independent of $\mathcal{N}(\mathbf{B}_y)$. For a positive integer $N$ the uniform $N\times N$
grid   on the unit square is defined
by
\begin{equation}
\mathbf{I}=\left\{ \frac{1}{N},\frac{2}{N},\cdots ,\frac{N-1}{N},1\right\}
^{2}.  \label{def I}
\end{equation}%
Each element $x$ of the grid $\mathbf{I}$ is called pixel. The number of pixels
is $n=N^{2}.$  Suppose that $%
x=(x^{(1)},x^{(2)})\in \mathbf{I}$,
and $\mathbf{B}_x=(x^{(1)}-1/N,x^{(1)}]\times(x^{(2)}-1/N,x^{(2)}]$. Then $\{%
\mathbf{B}_x\}_{x\in \mathbf{I}}$ is a partition of the square $(0,1]\times(0,1]$. The image function $f$ is considered to be constant on each
 $\mathbf{B}_x$, $x\in \mathbf{I}$.
Hence we get a discrete function $f(x)=\Lambda(\mathbf{B}_x)$, $%
x\in \mathbf{I}$. The denoising aims at estimating the underlying intensity
profile $f(x)$.
In the sequence we shall use the following important property of the Poisson distribution:
\begin{equation}
\mathbb{E}(\mathcal{N}(\mathbf{B}_{x}))= \mathbb{V}ar(
\mathcal{N}(\mathbf{B}_{x}))=f(x).
\label{porperty Poisson}
\end{equation}
Actually the Poisson noise model can be viewed as the following additive noise model
\begin{equation}
Y(x)=f(x) + \epsilon(x),
\label{model poisson2}
\end{equation}
where
\begin{equation}
\epsilon (y)=Y(x)-f(x).
\label{defi epsilon}
\end{equation}
may be considered as an additive heteroscedastic   noise related to the Poisson model. Due to
(\ref{porperty Poisson}), we have
  $\E (\epsilon(y))=0$ and $\mathbb{V}ar(\epsilon(y))=\mathbb{V}ar(Y(y))=f(x)$.

Let us set some  notations to be used
throughout the paper. The Euclidean norm of a vector $x=\left(x_{1},...,x_{d}\right)
\in \mathbf{R}^{d}$ is denoted by $%
\left\Vert x\right\Vert _{2}=\left(\sum_{i=1}^{d}x_{i}^{2}\right)^{\frac{1}{2}} .$ The
supremum norm of $x$ is denoted by $\Vert x\Vert _{\infty }=\sup_{1\leq
i\leq d}\left\vert x_{i}\right\vert .$ The cardinality of a set $\mathbf{A}$ is
denoted $\mathrm{card}\, \, \mathbf{A}$.  For any pixel $x_{0}\in \mathbf{I}$ and a given $h>0,$ the
square window
\begin{equation}
\mathbf{U}_{x_{0},h}=\left\{ x\in \mathbf{I:\;}\Vert x-x_{0}\Vert _{\infty
}\leq h\right\}  \label{def search window}
\end{equation}%
is called \emph{search window} at $x_{0}.$ We naturally take $h$ as a multiple of $\frac{1}{N}$ ($ h=\frac{k}{N}$ for some $k\in \{ 1, 2,\cdots,N\}$). The size of the square
search window $\mathbf{U}_{x_{0},h}$ is the positive integer number
\begin{equation*}
M=(2Nh+1)^{2}=\mathrm{card}\, \ \mathbf{U}_{x_{0},h}.
\end{equation*}
 For any pixel $x\in \mathbf{U}%
_{x_{0},h}$ and a given $\eta >0$. Consider a second square window $\mathbf{U}_{x,\eta}$ of size
\begin{equation*}
m=(2N\eta+1)
^{2}=\mathrm{card\ }\mathbf{U}_{x_{0},\eta }.
\end{equation*}
We shall call $\mathbf{U}_{x,\eta}$ local patches and $\mathbf{U}_{x,h}$ search windows.
 Finally, the positive part of  a real number $%
a$ is denoted by $a^{+}$:
\begin{equation*}
a^{+}=\left\{
\begin{array}{cc}
a & \text{if }a\geq 0, \\
0 & \text{if }a<0.%
\end{array}%
\right.
\end{equation*}

\subsection{Construction of the estimator}

Let $h>0$ be fixed. For any pixel $x_{0}\in \mathbf{I}$ consider a family of
weighted estimates $\widetilde{f}_{h,w}(x_{0})$ of the form
\begin{equation}
\widetilde{f}_{h,w}(x_{0})=\sum_{x\in \mathbf{U}_{x_{0},h}}w(x)Y(x),
\label{s2fx}
\end{equation}%
where the unknown weights satisfy
\begin{equation}
w(x)\geq 0\;\;\;\text{and\ \ \ }\sum_{x\in \mathbf{U}_{x_{0},h}}w(x)=1.
\label{s2wx}
\end{equation}%
The usual bias and  variance decomposition of the Mean Square Error
gives%
\begin{equation}
\mathbb{E}\left(\widetilde{f}_{h,w}(x_{0})-f(x_{0})\right)
^{2}=Bias^{2}+Var,  \label{s2ef}
\end{equation}%
where
\begin{equation*}
Bias^{2}=\left(\sum_{x\in \mathbf{U}_{x_{0},h}}w(x)\left(f(x)-f(x_{0})\right) \right) ^{2}
\end{equation*}
and
\begin{equation*}
Var=\sum_{x\in
\mathbf{U}_{x_{0},h}}w(x)^{2}f(x).
\end{equation*}%
The decomposition (\ref{s2ef}) is commonly used to construct asymptotically
minimax estimators over some given classes of functions in the nonparametric
function estimation.
With our approach the bias term $Bias^{2}$ will be bounded in terms of the
unknown function $f$ itself. As a result we obtain some "oracle" weights $w$
adapted to the unknown function $f$ at hand, which will be estimated
further using data patches from the image $Y.$

First, we shall address the problem of determining the "oracle" weights. With
this aim denote%
\begin{equation}
\rho(x)=\rho _{f,x_{0}}\left(x\right) = \left\vert f(x)-f(x_{0})\right\vert .
\label{s2rx}
\end{equation}%
Note that the value of $\rho _{f,x_{0}}\left(x\right) $ characterizes the
variation of the image brightness of the pixel $x$ with respect to the pixel
$x_{0}.$ From the decomposition (\ref{s2ef}), we easily obtain a tight upper
bound in terms of the vector $\rho _{f,x_{0}}:$%
\begin{equation}
\mathbb{E}\left(\widetilde{f}_{h}(x_{0})-f(x_{0})\right) ^{2}\leq g_{\rho
_{f,x_{0}}}(w)=g_{\rho}(w),  \label{s2ef-vers2}
\end{equation}%
where
\begin{equation}
g_{\rho }(w)=\left(\sum_{x\in \mathbf{U}_{x_{0},h}}w(x)\rho (x)\right)
^{2}+\sum_{x\in \mathbf{U}_{x_{0},h}}w(x)^{2}f(x).
\label{def gw}
\end{equation}

From the following theorem we can obtain the form of the weights $w$ which
minimize the function $g_{\rho }(w)$ under the constraints (\ref%
{s2wx}) in terms of  $\rho \left(x\right) .$
For the sake of generality, we shall formulate the result for an arbitrary non-negative
function $\rho(x)$, $x\in \mathbf{U}_{x,h}$, not necessarily defined by (\ref{s2rx}).

Introduce into consideration the strictly increasing function%
\begin{equation}
M_{\rho }\left(t\right) =\sum_{x\in \mathbf{U}_{x_{0},h}}\frac{1}{f(x)}\rho (x)(t-\rho
(x))^{+},\ \ \ t\geq 0.
\label{def mt}
\end{equation}%
Let $K_{\text{tr}}$ be the usual triangular kernel:
\begin{equation}
K_{\text{tr}}\left(t\right) =\left(1-\left\vert t\right\vert \right)
^{+},\quad t\in \mathbf{R}^{1}.
\label{def kernel}
\end{equation}

\begin{theorem}
\label{Th weights 001}Let $\rho \left(x\right) ,$ $x\in \mathbf{U}%
_{x_{0},h} $ be an arbitrary similarity function and let $g_{\rho}(w)$ be given by (%
\ref{def gw}). Suppose that $f(x)> 0$ for all $x\in \mathbf{U}_{x_0,h}$. Then there are  unique weights which minimize $g_{\rho}(w)$
subject to (\ref{s2wx}), given by
\begin{equation}
w_{\rho }(x)=\frac{%
K_{\text{tr}}\left(\frac{\rho(x)}{a}\right)/f(x)}{\sum_{y\in \mathbf{U}_{x_0,h}}
K_{\text{tr}}\left(\frac{\rho(y)}{a}\right)/f(y)},  \label{eq th weights 001}
\end{equation}%
where $a>0$ is the unique solution of the equation
\begin{equation}
M_{\rho }\left(a\right) =1.
\label{s2rt}
\end{equation}
\end{theorem}

 The proof of Theorem  \ref{Th weights 001} is deferred to Section  \ref{Sec: proof of Th weights
001}.

\begin{remark}
\label{calculate a} The bandwidth $a>0$ is the solution of
\begin{equation*}
\sum_{x\in \mathbf{U}_{x_{0},h}}\frac{1}{f(x)}\rho(x)({a} -%
\rho(x))^{+}=1,
\end{equation*}%
and can be calculated as follows. We sort the set $\{\rho(x)\,|\,x\in \mathbf{U}_{x_{0},h}\}$ in the ascending order $0=%
\rho_{1}\leq \rho_{2}\leq \cdots \leq \rho%
_{M}<\rho_{M+1}=+\infty $, where $M=\mathrm{card}\, \,\mathbf{U}_{x_{0},h}$.
 Let $f_i$ be the corresponding value of $f(x)$
  (we have $f_i=f(x)$ if $\rho_i=\rho(x)$, $x\in \mathbf{U}_{x_0,h}$).
Let
\begin{equation}
a_{k}=\frac{1+\sum\limits_{i=1}^{k}\rho_{i}^{2}/f_i}{%
\sum\limits_{i=1}^{k}\rho_{i}/f_i},\quad 1\leq k\leq M,  \label{a k}
\end{equation}%
and
\begin{eqnarray}
k^{\ast } &=&\max \{1\leq k\leq M\,|\,a_{k}\geq \rho_{k}\}  \notag
\\
&=&\min \{1\leq k\leq M\,|\,a_{k}<\rho_{k}\}-1,  \label{k star}
\end{eqnarray}%
with the convention that $a_{k}=\infty $ if $\rho_{k}=0$ and that
$\min \varnothing =M+1$. The bandwidth $a>0$ can be expressed as $ {a} %
=a_{k^{\ast }}$. Moreover, $k^{\ast }$ is also the unique integer $k\in
\{1,\cdots ,M\}$ such that $a_{k}\geq \rho_{k}$ and $a_{k+1}<%
\rho_{k+1}$ if $k<M$.
\end{remark}

The proof of Remark \ref{calculate a} can be found in \citep{JinGramaLiuowf}.

Let $\rho \left(x\right) ,$ $x\in \mathbf{U}_{x_{0},h},$ be an arbitrary
non-negative function and let $w_{\rho }$ be the optimal weights given by
(\ref{eq th weights 001}). Using these weights $w_{\rho }$
we define the family of estimates%
\begin{equation}
f_{h}^{\ast }(x_{0})=\sum_{x\in \mathbf{U}_{x_{0},h}}w_{\rho }(x)Y(x)
\label{s2fx3}
\end{equation}%
depending on the unknown function $\rho .$ The next theorem shows that one
can pick up an useful estimate from the family $f_{h}^{\ast }$ if the
function $\rho $ is close to the "true" function $\rho _{f,x_{0}}(x)=\left\vert
f\left(x\right) -f\left(x_{0}\right) \right\vert ,$ i.e. if
\begin{equation}
\rho \left(x\right) =\left\vert f\left(x\right) -f\left(x_{0}\right)
\right\vert +\delta _{n},  \label{s2fx3a}
\end{equation}%
where $\delta _{n}\geq 0$ is a small deterministic error. We shall prove the
convergence of the estimate $f_{h}^{\ast }$ under the local H\"{o}lder
condition
\begin{equation}
|f(x)-f(y)|\leq L\Vert x-y\Vert _{\infty }^{\beta },\,\,\,\forall x,\,y\in
\mathbf{U}_{x_{0},h+\eta},  \label{Local Holder cond}
\end{equation}%
where $\beta >0$ is a constant, $h>0,$ $\eta>0$ and $x_{0}\in \mathbf{I}.$
\par
In the
following, $c_{i}>0$ $(i\geq 1)$ denotes a positive constant,
and $O(a_n)$  $(n\geq 1)$ denotes a sequence  bounded
by $c\cdot a_n$ for some constant $c>0$ and all $n\geq 1$.
 All the constants $c_i>0$ and $c>0$ depend only on $L$
 and  $\beta$; their values can be different from line to line. Let
 \begin{equation}
 \Gamma \geq \max\{f(x):x\in \mathbf{I}\}
\label{defi Gamma}
 \end{equation}
 be an upper bound of the image $f$.

\begin{theorem}
\label{Th oracle 001}Assume  that $h\geq
c_{0}n^{-\alpha }$ with $0\leq \alpha <\frac{1}{2\beta +2}$ and $c_{0}>0,$
or that $h=c_{0}n^{-\frac{1}{2\beta +2}}$ with $c_{0}>c_{1}=\left(\Gamma
\frac{\left(\beta +2\right) \left(2\beta +2\right) }{8L^{2}\beta }\right)
^{\frac{1}{2\beta +2}}.$  Suppose also that the function $f>0$ satisfies the local H%
\"{o}lder condition (\ref{Local Holder cond}). Let $f_{h}^{\ast }(x_{0})$ be given by (\ref{s2fx3}%
), where the weights $w_{\rho}$ are defined by (\ref{eq th weights 001}) and (\ref{s2rt}%
) with $\rho \left(x\right) =\left\vert f\left(x\right) -f\left(x_{0}\right) \right\vert +\delta _{n}$ and $\delta _{n}=O\left(n^{-\frac{%
\beta }{2+2\beta }}\right) .$ Then
\begin{equation}
\mathbb{E}\left(f_{h}^{\ast }(x_{0})-f(x_{0})\right) ^{2}=O\left(n^{-\frac{%
2\beta }{2+2\beta }}\right) .  \label{s2ef2}
\end{equation}
\end{theorem}

For the proof of this theorem see Section \ref{Sec: proof of Th oracle 001}.

Recall that the bandwidth $h$ of order $n^{-\frac{1}{2+2\beta }}$ is
required to have the optimal minimax rate of convergence $O\left(n^{-\frac{2\beta }{%
2+2\beta }}\right) $ of the Mean Square Error for estimating the function $f
$ of local H\"{o}lder smoothness $\beta $  (cf. e.g.
 \cite{FanGijbels1996}). To better understand the adaptivity property of
the oracle $f_{h}^{\ast }(x_{0}),$ assume that the image $f$ at $x_{0}$ has local
H\"{o}lder smoothness $\beta $ (see \citep{Wh}) and that $h\geq c_{0}n^{-\alpha }$ with $0\leq
\alpha <\frac{1}{2\beta +2},$ which means that the radius $h>0$ of the search
window $U_{x_{0},h}$ has been chosen larger than the ``standard" $n^{-\frac{1%
}{2\beta +2}}.$ Then, by Theorem \ref{Th oracle 001}, the rate of convergence
of the oracle is still of order $n^{-\frac{\beta }{2+2\beta }}$. If we choose a sufficiently large
search window $U_{x_{0},h},$ then the oracle $f_{h}^{\ast }(x_{0})$ will
have a rate of convergence which depends only on the unknown maximal local
smoothness $\beta $ of the image $f.$ In particular, if $\beta $ is very
large, then the rate will be close to $n^{-1/2},$ which ensures a good
estimation of the flat regions in cases where the regions are indeed flat.
More generally, since Theorem \ref{Th oracle 001} is valid for arbitrary $%
\beta ,$ it applies for the maximal local H\"{o}lder smoothness $\beta _{x_{0}}$
at $x_{0},$ therefore the oracle $f_{h}^{\ast }(x_{0})$ will exhibit the
best rate of convergence of order $n^{-\frac{2\beta _{x_{0}}}{2+2\beta
_{x_{0}}}}$ at $x_{0}.$ In other words, the procedure adapts to the best
rate of convergence at each point $x_{0}$ of the image.

We justify by simulation results that the difference
between the oracle $f_{h}^{\ast }$ computed with $\rho(x) =\rho
_{f,x_{0}}(x)=\left\vert f\left(x\right) -f\left(x_{0}\right) \right\vert ,$
and the true image $f$, is extremely small (see Table\ \ref{Table oracle}).
This shows that, at least from the practical point of view, it is justified
to optimize the upper bound $g_{\rho _{f,x_{0}}}(w)$ instead of optimizing
the Mean Square Error $\mathbb{E}\left(f^{\ast }_h(x_{0})-f(x_{0})\right) ^{2}$ itself.

The estimate $f_{h}^*$ with the choice $\rho \left(x\right) =\rho_{f,x_0}\left(x\right) $ will be called oracle filter. In
particular for the oracle filter $f_{h}^{\ast },$ under the conditions of
Theorem \ref{Th oracle 001}, we have%
\begin{equation*}
\mathbb{E}\left(f_{h}^{\ast }(x_{0})-f(x_{0})\right) ^{2}\leq g_{\rho}\left(w_{\rho}\right) \leq cn^{-\frac{2\beta }{2+2\beta }}.
\end{equation*}

Now, we turn to the study of the convergence of the Optimal Weights Filter.
 Due to the difficulty in dealing with the dependence of
  the weights we shall consider a slightly modified version
   of the proposed algorithm:  we divide  the set of pixels
   into two disjoint parts,  so that the weights are constructed from one part,  and  the estimation of the target function is a weighted mean  along  the other part. More precisely,
   we proceed as follows.
Assume that $x_{0}\in \mathbf{I}$. Denote
\begin{equation*}
\mathbf{I}'_{x_{0}}=\left\{ x_{0}+\left(\frac{i}{N},\frac{j}{N}\right) \in
\mathbf{I}:i+j\text{ is pair }\right\} ,
\end{equation*}%
and $\mathbf{I}''_{x_{0}}=\mathbf{I}\diagdown \mathbf{I}'_{x_{0}}.$
Denote $\mathbf{U}_{x_{0},h}^{\prime }=\mathbf{U}_{x_{0},h}\cap \mathbf{I}'_{x_{0}}$ and $\mathbf{U}_{x,\eta }^{\prime\prime }=\mathbf{U}%
_{x,\eta }\cap \mathbf{I}''_{x_{0}}.$
Since $%
\mathbf{E}|Y(x)-Y(x_{0})|^{2}=|f(x)-f(x_{0})|^{2}+f(x_0)+f(x)$, an obvious
estimate of $\mathbf{E}\left\vert Y(x)-Y(x_{0})\right\vert ^{2}$ is given by
\begin{equation*}
\frac{1}{\mathrm{card}\, \mathbf{U}''_{x_{0},\eta }}\sum_{y\in \mathbf{U}''_{x_{0},\eta
}}\left\vert Y(y)-Y(Ty)\right\vert ^{2},
\end{equation*}%
where $T=T_{x_{0},x}$ is the translation mapping: $Ty=x+(y-x_{0})$.
Define an estimated similarity function $\widehat{\rho }_{x_0}$ by
\begin{equation}
\widehat{\rho }_{x_0}(x)=\left(\left(\frac{1}{\mathrm{card}\, \mathbf{U}''_{x_{0},\eta }}\sum_{y\in {\mathbf{U}''_{x_{0},\eta }}}|Y(y)-Y(T
y)|^{2}\right) ^{1/2}-\sqrt{2\overline{f}(x_0)} \right)^+,  \label{s3rx}
\end{equation}%
where
\begin{equation*}
\overline{f}(x_0)=\frac{1}{\mathrm{card}\, \mathbf{U}''_{x_{0},\eta }}\sum_{y\in {\mathbf{U''}%
_{x_{0},h }}}Y(y).
\end{equation*}
The Optimal Weights Poisson Noise Filter (OWPNF) proposed in this paper is defined by
\begin{equation}
\widehat{f}_{h }(x_{0})=\sum_{x\in \mathbf{U}_{x_{0},h}^{\prime }}%
\widehat{w}(x)Y(x),  \label{s3fh}
\end{equation}%
where
\begin{equation}
\widehat{w}=\arg \min_{w}\left(\sum_{x\in \mathbf{U}_{x_{0},h}^{\prime
}}w(x)\widehat{\rho }_{x_0}(x,x_{0})\right) ^{2}+\overline{f}(x_0)\sum_{x\in
\mathbf{U}_{x_{0},h}^{\prime }}w^{2}(x).  \label{s3ww}
\end{equation}

In the next theorem, we prove that with the choice $h=c_{0}n^{-\frac{1}{%
2\beta +2}}$ and $\eta =c_{2}n^{-\frac{1}{2\beta +2}},$ the Mean Square Error of the estimator $\widehat{f}_{h}(x_{0})$ converges nearly at the rate
$n^{-\frac{2\beta }{2\beta +2}}$ which is the usual optimal rate of
convergence for a given H\"{o}lder smoothness $\beta >0$ (see e.g. Fan and
Gijbels (1996 \citep{FanGijbels1996})).

\begin{theorem}
\label{Th adapt 001}Assume that
$h=c_{0}n^{-\frac{1}{2\beta +2}}$ with $c_{0}>c_{1}=\left(\Gamma \frac{%
\left(\beta +2\right) \left(2\beta +2\right) }{8L^{2}\beta }\right) ^{%
\frac{1}{2\beta +2}}$, and that $\eta =c_{2}n^{-\frac{1}{2\beta +2}}.$ Suppose that the function $f(x)\geq \frac{1}{\ln n}$
satisfies the local H\"{o}lder condition (\ref{Local Holder cond}).  Then
\begin{equation}
\mathbb{E}(\widehat{f}_{h}(x_{0})-f(x_{0}))^{2}=O\left(n^{-\frac{2\beta }{%
2\beta +2}}\ln ^{2}n\right) .  \label{s3ef}
\end{equation}
\end{theorem}

For the proof of this theorem see Section \ref{Sec: proof of Th adapt 001}.

\section{\label{Sec:simulations}Simulation}

For simulations we use the following usual set of $256\times 256$ images:
Spots$[0.08,4.99]$, Galaxy$[0,5]$, Ridges$[0.05,0.85]$, Barbara$[0.93,15.73]$
and Cells
$[0.53,16.93]$ (see the first row of Figure \ref{Fig oracle}). All
the images are included in the package "Denoising software for Poisson data"
which can be downloaded at http://www.cs.tut.fi/~foi/invansc/. We first do
the simulations with the oracle filter which shows excellent visual quality
of the reconstructed image. We next present our denoising algorithm and the
numerical results which are compared with related recent works (%
\citep{makitalo2011optimal} and \citep{ZHANG2008WAVELETS}). Each of the
aforementioned articles proposes an algorithm specifically designed for
Poisson noise removal (EUI+BM3D, MS-VST + $7/9$ and MS-VST + B3
respectively).

We evaluate the performance of a denoising filter $\widehat{f}$ by using the
Normalized Mean Integrated Square Error (NMISE) defined by
\begin{equation*}
NMISE=\frac{1}{n^{\ast }}\sum_{f(x)>0,x\in \mathbf{I}}\left(\frac{(\widehat{f}%
(x)-f(x))^{2}}{f(x)}\right) ,
\end{equation*}%
where $\widehat{f}(x)$ are the estimated intensities, $f(x)$ are the
respective true vales, and $n^{\ast }=\mathrm{card}\, \{f(x):f(x)>0,x\in \mathbf{I}\}$.

\subsection{Oracle Filter}

In this section we present the denoising algorithm called Oracle Filter, and
show its performance on some test images.

\noindent\rule{\textwidth}{.2pt}

\textbf{Algorithm:}\quad Oracle Filter

\noindent\rule{\textwidth}{.2pt}

Repeat for each $x_{0}\in \mathbf{I}$

\quad \quad Let $a=1$ (give the initial value of a)

\quad \quad compute$\;{\rho }(x_{i})$ by (\ref{empir simil func})

\quad \quad reorder ${\rho }(x_{i})$ as increasing sequence

\quad \quad loop from $k=1$ to $M$

\quad \quad \quad  if $\sum_{i=1}^{k}{\rho }(x_{i})>0$

\quad \quad \quad \quad   if $\frac{1+\sum_{i=1}^{k}{\rho }^{2}(x_{i})/f(x_i)}{\sum_{i=1}^{k}\rho (x_{i})/f(x_i)}\geq {\rho }(x_{k})$ then $a=\frac{1+\sum_{i=1}^{k}{\rho }^{2}(x_{i})/f(x_i)}{\sum_{i=1}^{k}\rho (x_{i})/f(x_i)}\geq {\rho }(x_{k})$

\quad \quad \quad \quad  else quit loop

\quad \quad \quad  else continue loop

\quad \quad end loop

\quad \quad compute $w(x_{i})=\frac{(a-{\rho }(x_{i}))^{+}/f(x_i)}{\sum_{x_{i}\in
\mathbf{U}_{x_{0},h}}(a-{\rho }(x_{j}))^{+}/f(x_j)}$

\quad \quad compute $f_h^{\ast }(x_{0})=\sum_{x_{i}\in \mathbf{U}%
_{x_{0},h}}w(x_{i})Y(x_{i})$. \bigskip

\noindent\rule{\textwidth}{.2pt}

\noindent \rule{0pt}{0.2pt} We calculate the optimal weights from the
original image  and compute the oracle estimate from the observed image
contaminated by the Poisson noise. For choosing the convenient size of the
search windows, we do numerical experiments with different window sizes (see
Table \ref{Table oracle}). The  results show that the difference between
the oracle estimator $f_{h}^{\ast }$ and the true value $f$ is extremely
small. In Figure \ref{Fig oracle}, the second row illustrates the visual
quality of the restored images by the Oracle Filter with $M=19\times 19$. We can
see that almost all the details have been retained.

\begin{table}
\renewcommand{\arraystretch}{0.6} \vskip3mm {\fontsize{8pt}{\baselineskip}%
\selectfont
\caption{NMISE values when oracle estimator $f_{h}^{\ast }$ is applied with
different values of $M$.}
\label{Table oracle}
\begin{tabular}{l|rrrrrrr}
\hline
Size & $7\times7$ & $9\times9$ & $11\times11$ & $13\times13$ & $15\times15$
& $17\times17$ & $19\times19$ \\ \hline
Spots[0.08,4.99] & 0.0302 & 0.0197 & 0.0166 & 0.0139 & 0.0112 & 0.0098 &
0.0104 \\
Galaxy[0,5] & 0.0284 & 0.0208 & 0.0165 & 0.0144 & 0.0122 & 0.0107 & 0.0093
\\
Ridges[0.05,0.85] & 0.0239 & 0.0178 & 0.0131 & 0.0109 & 0.0098 & 0.0085 &
0.0074 \\
Barbara[0.93,15.73] & 0.0510 & 0.0399 & 0.0304 & 0.0248 & 0.0208 & 0.0195 &
0.0174 \\
Cells[0.53,16.93] & 0.0422 & 0.0323 & 0.0257 & 0.0216 & 0.0191 & 0.0164 &
0.0146 \\ \hline
\end{tabular}
} \vskip1mm
\end{table}

\begin{figure}[tbp]
\renewcommand{\arraystretch}{0.5} \addtolength{\tabcolsep}{-6pt} \vskip3mm {%
\fontsize{8pt}{\baselineskip}\selectfont
\begin{tabular}{ccccc}
\includegraphics[width=0.20\linewidth]{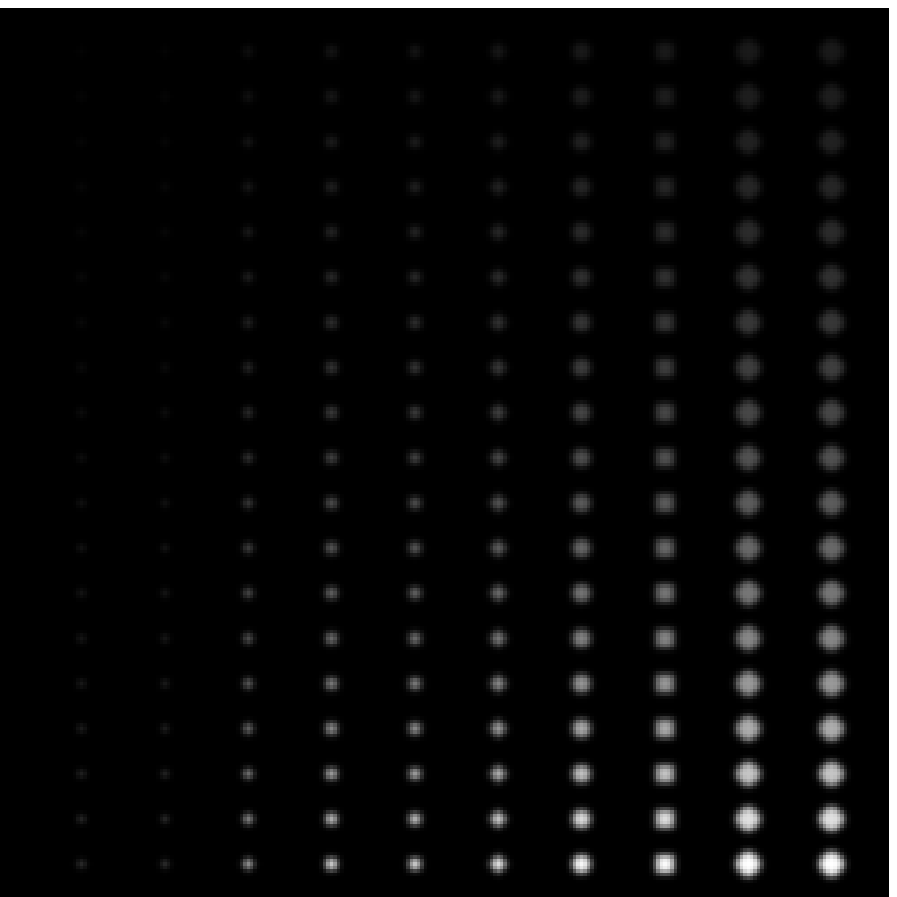} & %
\includegraphics[width=0.20\linewidth]{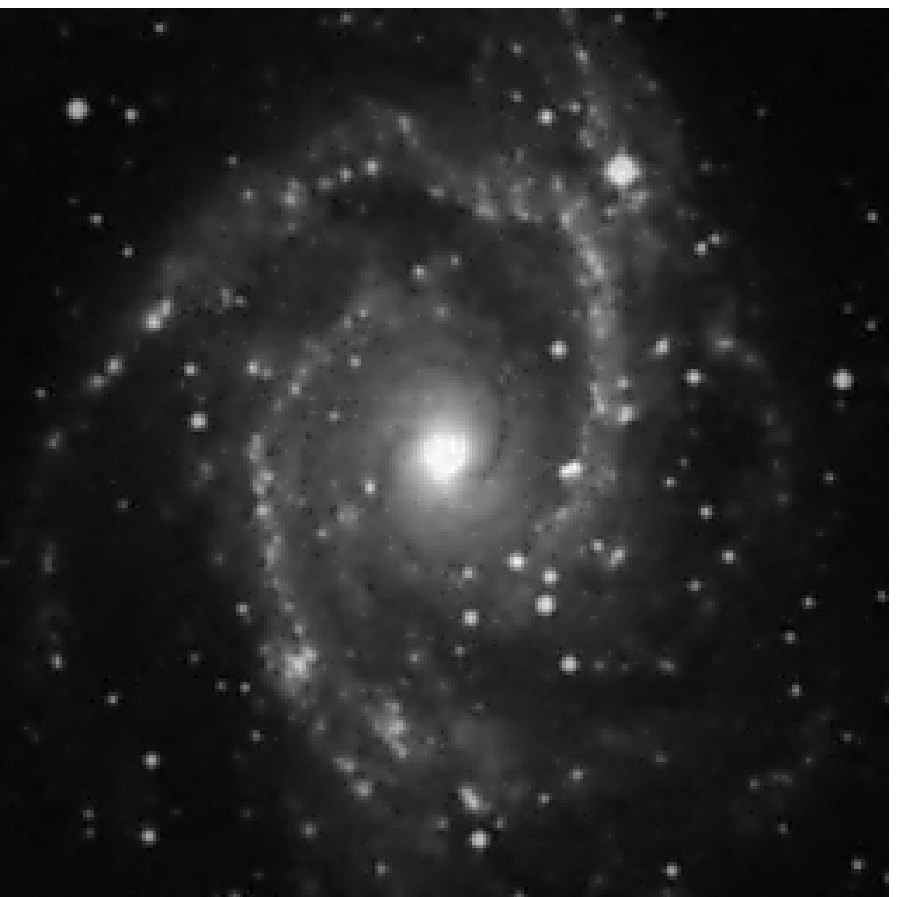} & %
\includegraphics[width=0.20\linewidth]{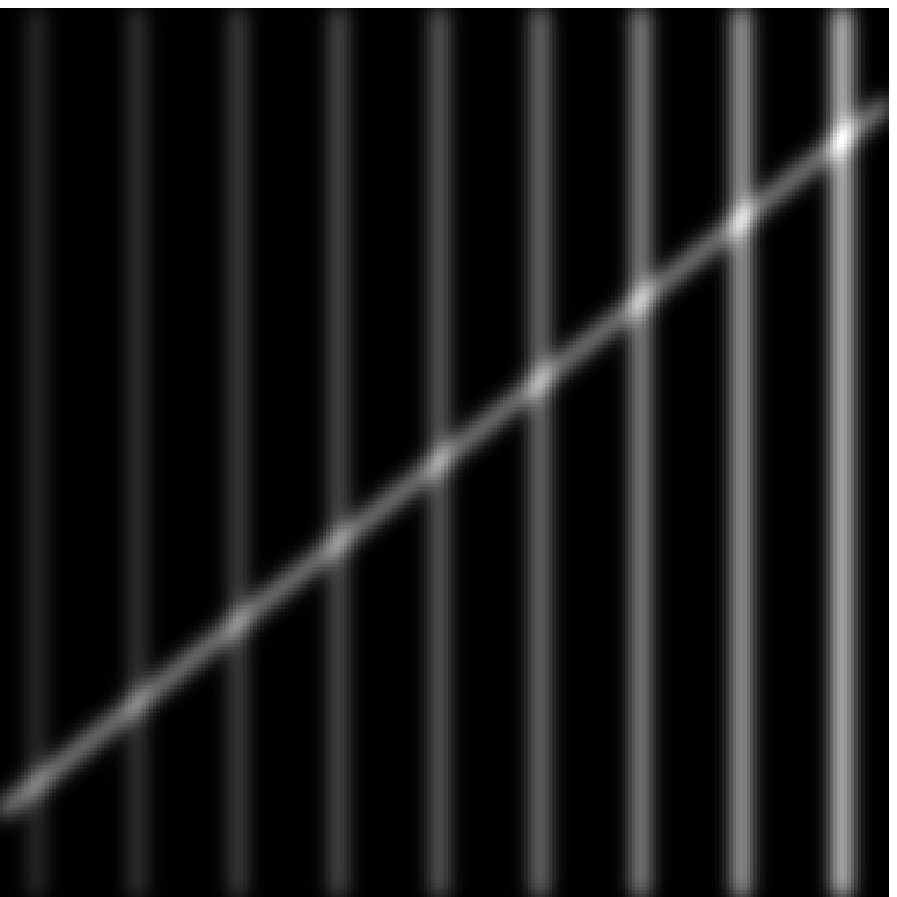} & %
\includegraphics[width=0.20\linewidth]{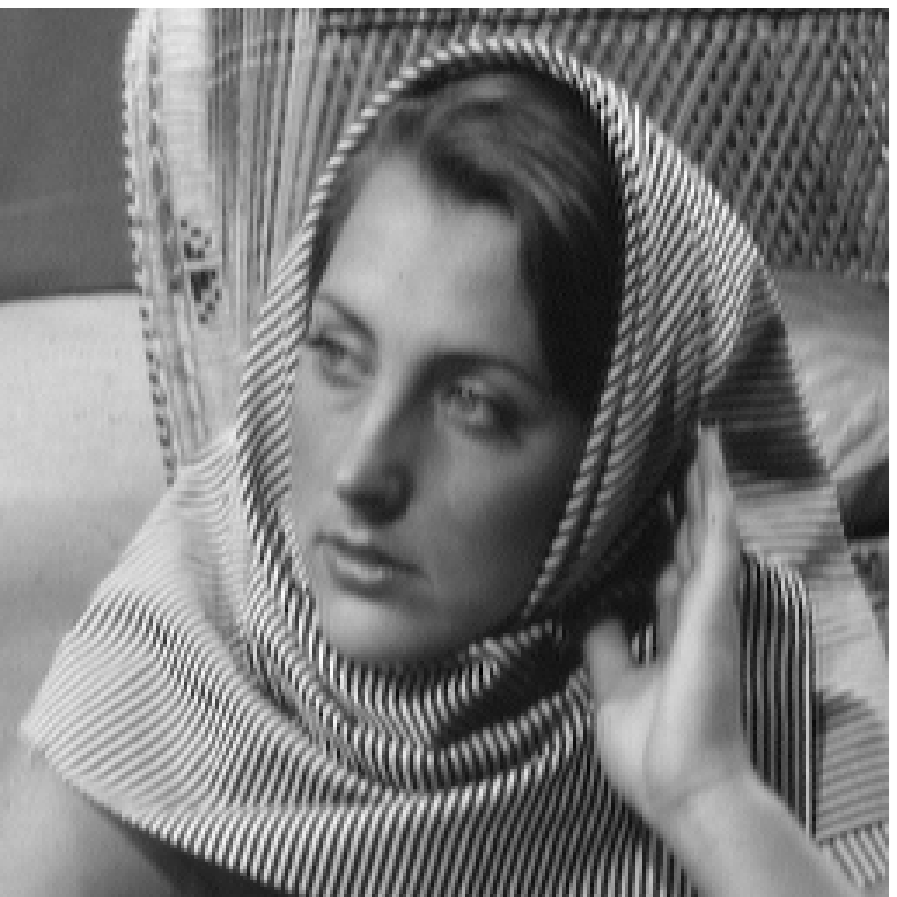} & %
\includegraphics[width=0.20\linewidth]{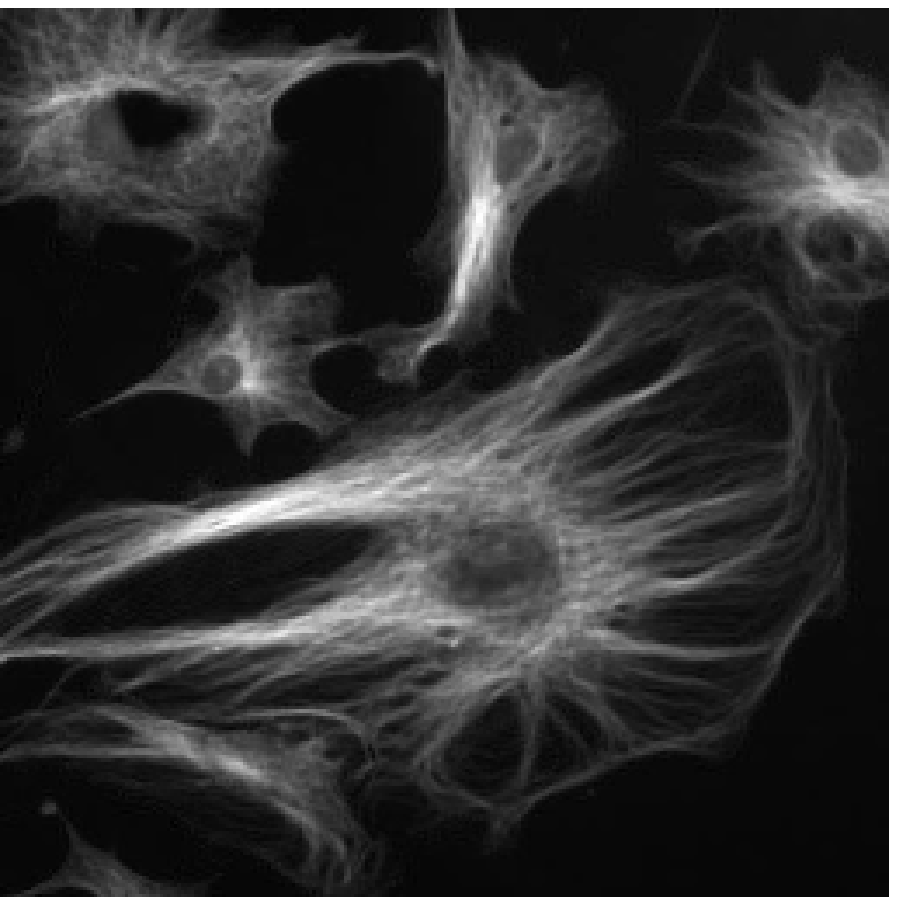} \\
\includegraphics[width=0.20\linewidth]{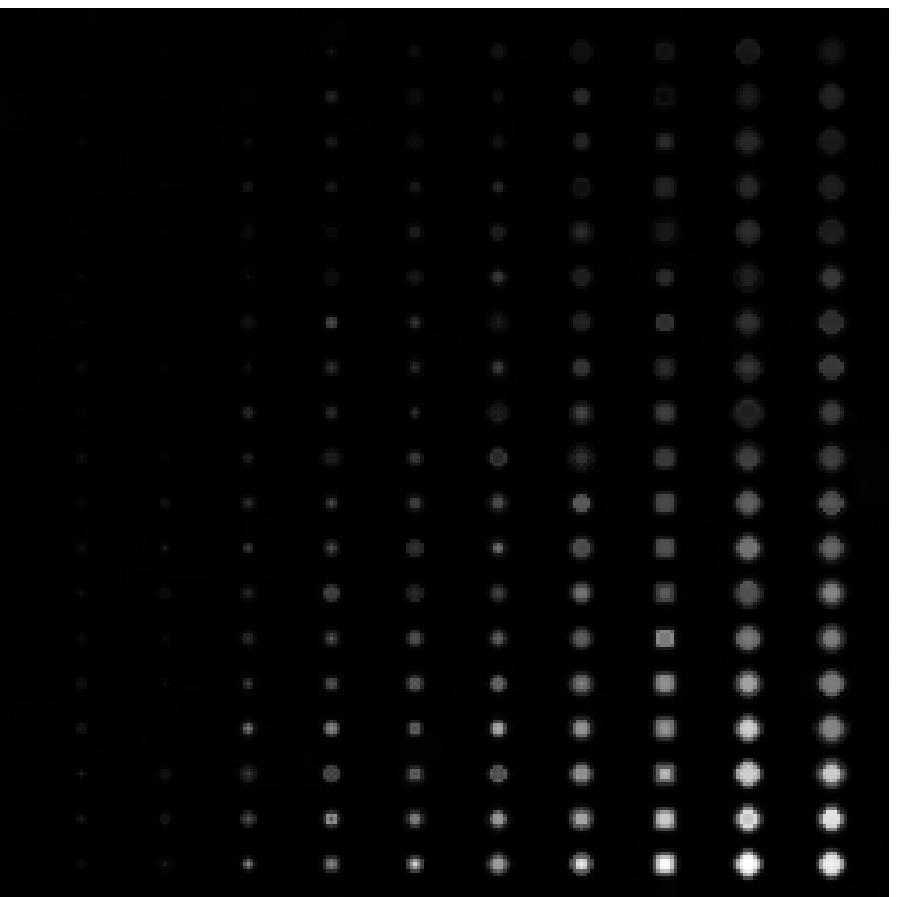} & %
\includegraphics[width=0.20\linewidth]{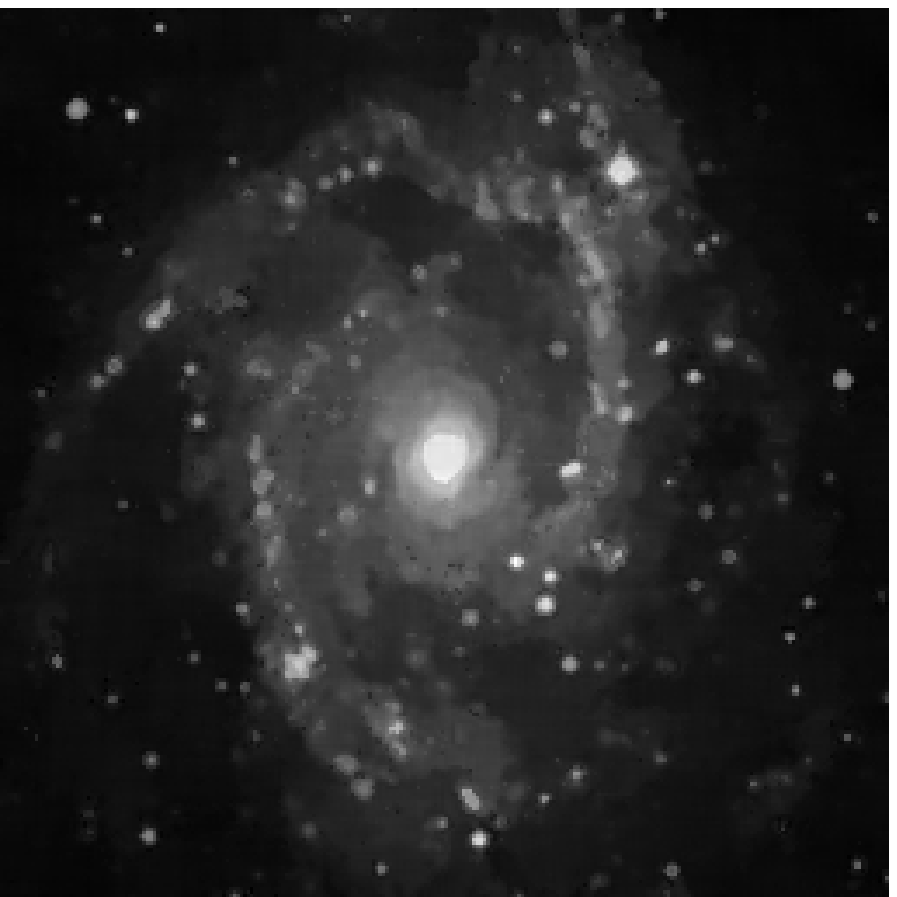} & %
\includegraphics[width=0.20\linewidth]{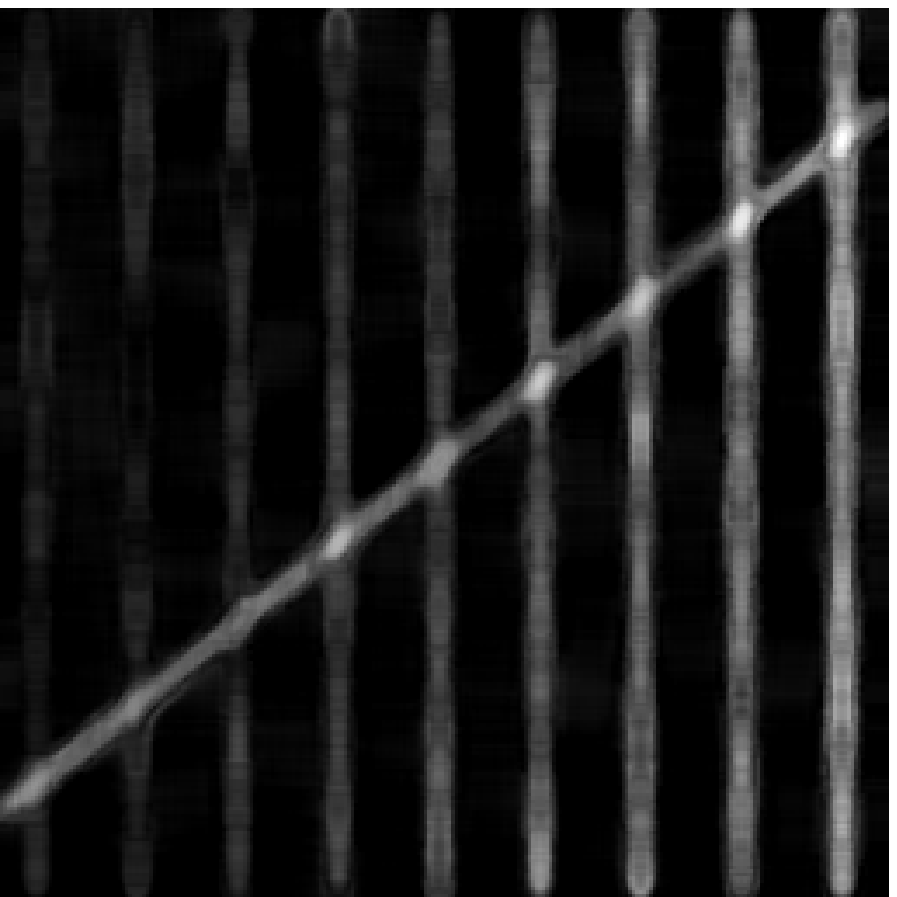} & %
\includegraphics[width=0.20\linewidth]{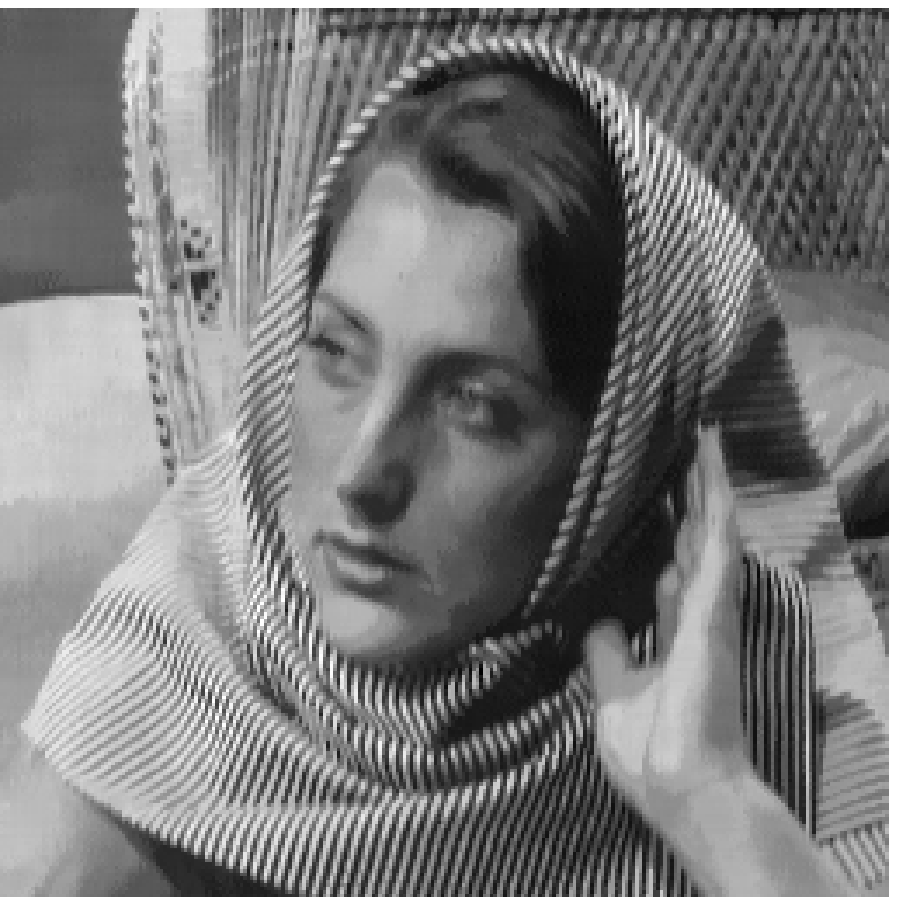} & %
\includegraphics[width=0.20\linewidth]{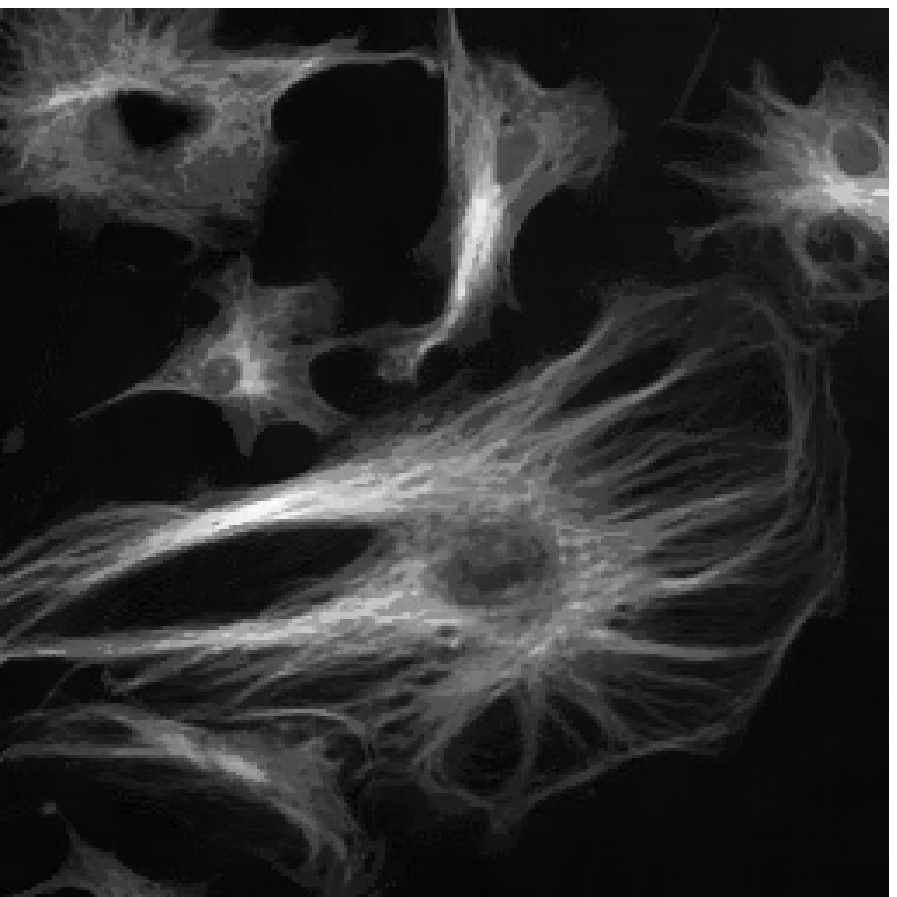} \\
(a) Spots & (b) Galaxy & (c) Ridges & (d) Barbara & (e) Cells%
\end{tabular}
}
\caption{{\protect\small The first row is the test images original, and the
second row is the images restored by Oracle Filter with $M=19\times19$. }}
\label{Fig oracle}
\end{figure}

\subsection{Performance of the Optimal Weights Poisson Noise Filter%
}

Throughout the simulations, we use the following algorithm for computing the
Optimal Weights Poisson Noise Filter $\widehat{f}_{h}(x_{0}).$ The input
values of the algorithm are $Y\left(x\right) ,$ $x\in \mathbf{I}$ (the
image) and two numbers $m=\left(2N\eta+1\right) \times \left(2N\eta+1\right) $ and
$M=\left(2Nh+1\right) \times \left(2Nh+1\right) $.
In order to improve the results we introduce a smoothed
version of the estimated similarity distance%
\begin{equation}
\widehat{\rho }_{\kappa,x_0}(x)=\left(\sqrt{\sum_{y\in
\mathbf{U}_{x_{0},\eta }}\kappa \left(y\right) \left\vert
(Y(y)-Y(Ty)\right\vert ^{2}}-\sqrt{2\overline{f}(x_{0})}\right) ^{+},
\label{empir simil func}
\end{equation}%
where
\begin{equation}
\kappa \left(y\right) =\frac{K(y)}{\sum_{y^{\prime }\in \mathbf{U}%
_{x_{0},\eta }}K(y^{\prime })}.
\end{equation}%
As smoothing kernels $K(y)$ we can use the Gaussian kernel
\begin{equation}
K_{g}(y,h)=\exp \left(-\frac{N^{2}\Vert y-x_{0}\Vert _{2}^{2}}{2h^{2}}%
\right) ,  \label{gaussian kernel}
\end{equation}%
the following kernel: for $y \in \mathbf{U}_{x_0,\eta}$,
\begin{equation}
K_{0}\left( y\right) =\sum_{k=\max(1,j)}^{N\eta}\frac{1}{(2k+1)^2}
 \label{kernel zero}
\end{equation}%
if $\|y-x_0\|_{\infty}=\frac{j}{N}$ for some $j\in \{0,1,\cdots,N\eta\}$,
 and  the
rectangular kernel%
\begin{equation}
K_{r}\left(y\right) =\left\{
\begin{array}{ll}
\frac{1}{\mathrm{card}\, \mathbf{U}_{x_{0},\eta }}, & y\in \mathbf{U}_{x_{0},\eta }, \\
0, & \text{otherwise.}%
\end{array}%
\right.  \label{rect kernel}
\end{equation}%
The best numerical results are obtained using $K(y) =K_{0}(y)$ in the
definition of $\widehat{\rho }_{\kappa,x_0}$. Also note that throughout the paper, we
symmetrize the image near the frontier.

We present below the denoising algorithm which realizes OWPNF and shows its performance on some test images.

\bigskip

\noindent\rule{\textwidth}{.2pt}

\textbf{Algorithm:}\quad Optimal Weights Poisson Noise Filter (OWPNF)

\noindent\rule{\textwidth}{.2pt}

\textbf{First step:}

Repeat for each $x_{0}\in \mathbf{I}$

\quad \quad Let $a=1$  (give the initial value of a)

\quad \quad compute$\;\widehat{\rho }_{\kappa,x_0}(x_{i})$ by (\ref{empir simil func})

\quad \quad reorder $\widehat{\rho }_{\kappa,x_0}(x_{i})$ as increasing sequence

\quad \quad loop from $k=1$ to $M$

\quad \quad \quad \quad if $\sum_{i=1}^{k}\widehat{\rho }_{\kappa,x_0}(x_{i})>0$

\quad \quad \quad \quad \quad \quad if $\frac{\overline{f}%
(x_0)+\sum_{i=1}^{k}\widehat{\rho }_{\kappa,x_0}^{2}(x_{i})}{\sum_{i=1}^{k}\widehat{%
\rho }_{\kappa}(x_{i})}\geq \widehat{\rho }_{\kappa,x_0}(x_{k})$ then $a=\frac{\overline{f}%
(x_0)+\sum_{i=1}^{k}\widehat{\rho }_{\kappa,x_0}^{2}(x_{i})}{\sum_{i=1}^{k}\widehat{%
\rho }_{\kappa}(x_{i})}$

\quad \quad \quad \quad \quad \quad else quit loop

\quad \quad \quad \quad else continue loop

\quad \quad end loop

\quad \quad compute $w(x_{i})=\frac{(a-\widehat{\rho }_{\kappa,x_0}(x_{i}))^{+}}{%
\sum_{x_{i}\in \mathbf{U}_{x_{0},h}}(a-\widehat{\rho }_{\kappa,x_0}(x_{j}))^{+}}$

\quad \quad compute $\widehat{f}^{\prime }(x_{0})=\sum_{x_{i}\in \mathbf{U}%
_{x_{0},h}}w(x_{i})Y(x_{i})$.

\textbf{Second step:}

For each $x_0\in \mathbf{I}$,
compute $\gamma(x_0)=\frac{1}{M}\sum_{x\in \mathbf{U}_{x_0,h}}\widehat{f}^{\prime }(x)$

\quad \quad If $\gamma(x_0)\leq 5$

\quad \quad compute $\widehat{f}(x_{0})= \frac{\sum _{\|x-x_0\|\leq d/N} K_{g}(x,H)%
\widehat{f}^{\prime }(x)}{\sum _{\|x-x_0\|\leq d/N} K_{g}(x,H)}$

\quad \quad else $\widehat{f}(x_{0})=\widehat{f}^{\prime }(x_0).$

\bigskip

\noindent\rule{\textwidth}{.2pt}

\begin{center}
\begin{figure}[tbp]
\renewcommand{\arraystretch}{0.5} \addtolength{\tabcolsep}{-6pt} \vskip3mm {%
\fontsize{8pt}{\baselineskip}\selectfont
\begin{tabular}{ccccc}
\includegraphics[width=0.20\linewidth]{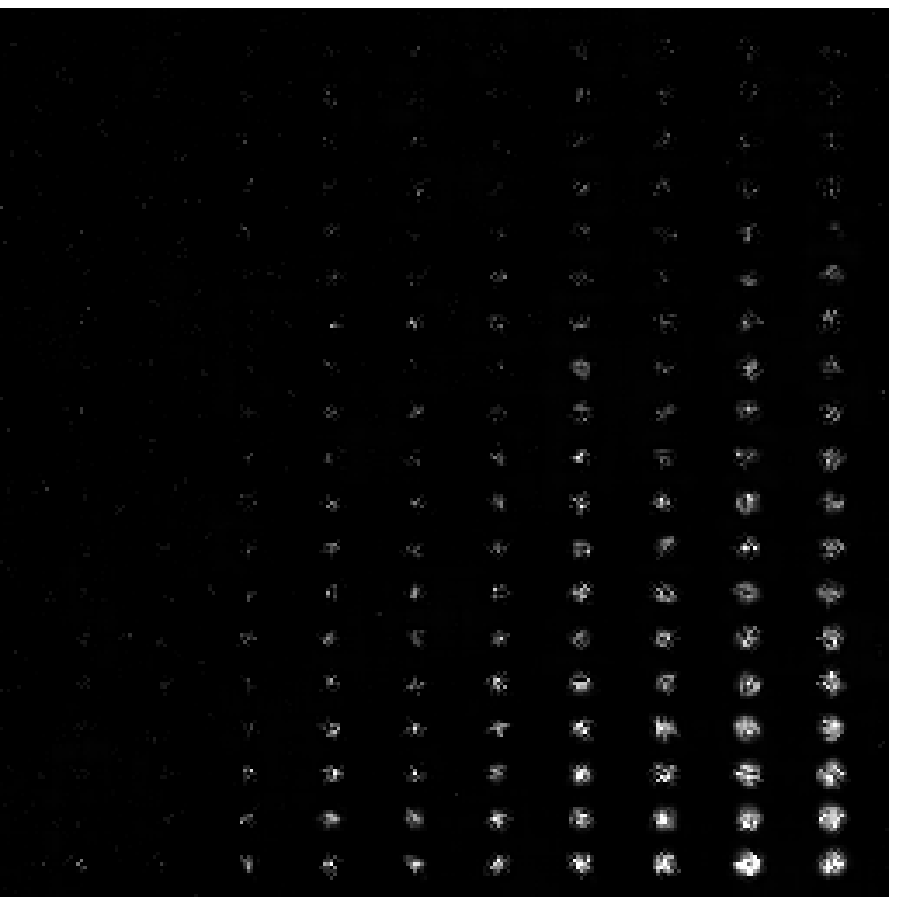} & %
\includegraphics[width=0.20\linewidth]{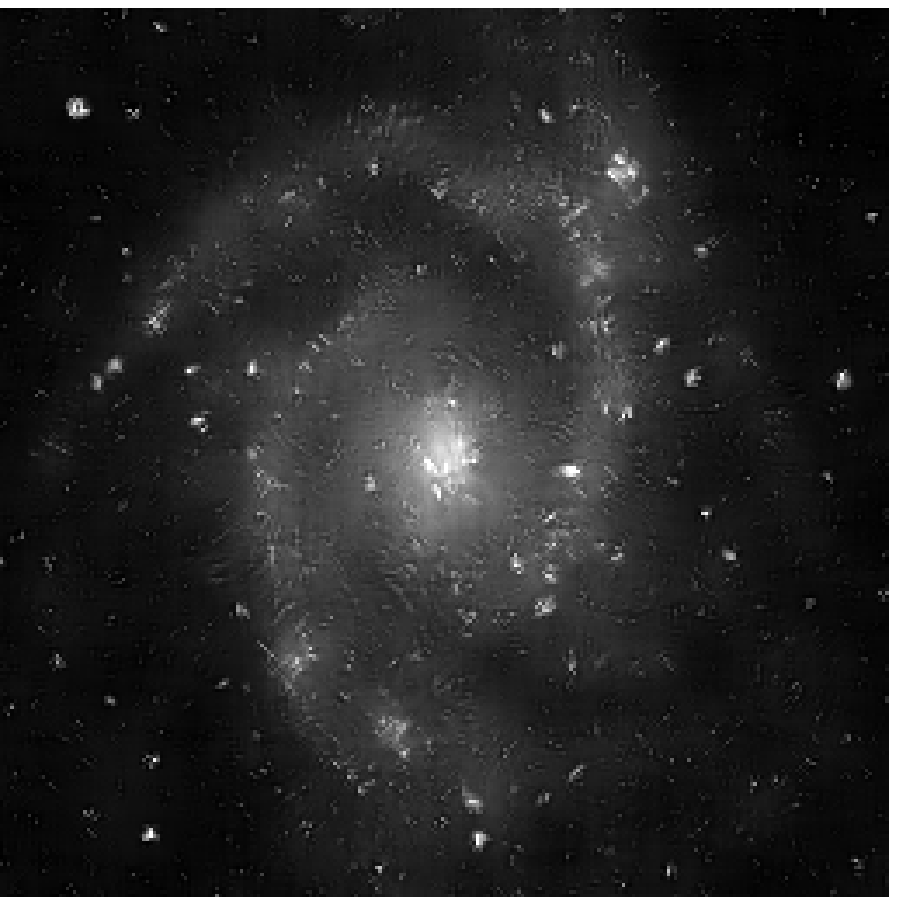} & %
\includegraphics[width=0.20\linewidth]{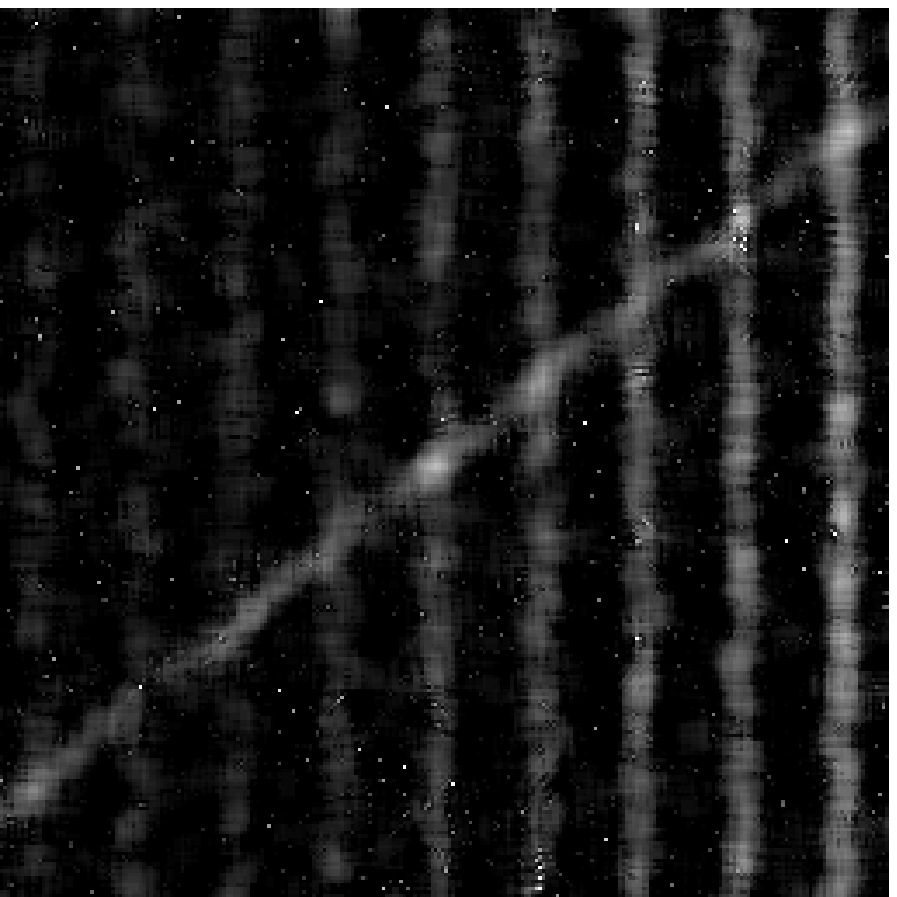} & %
\includegraphics[width=0.20\linewidth]{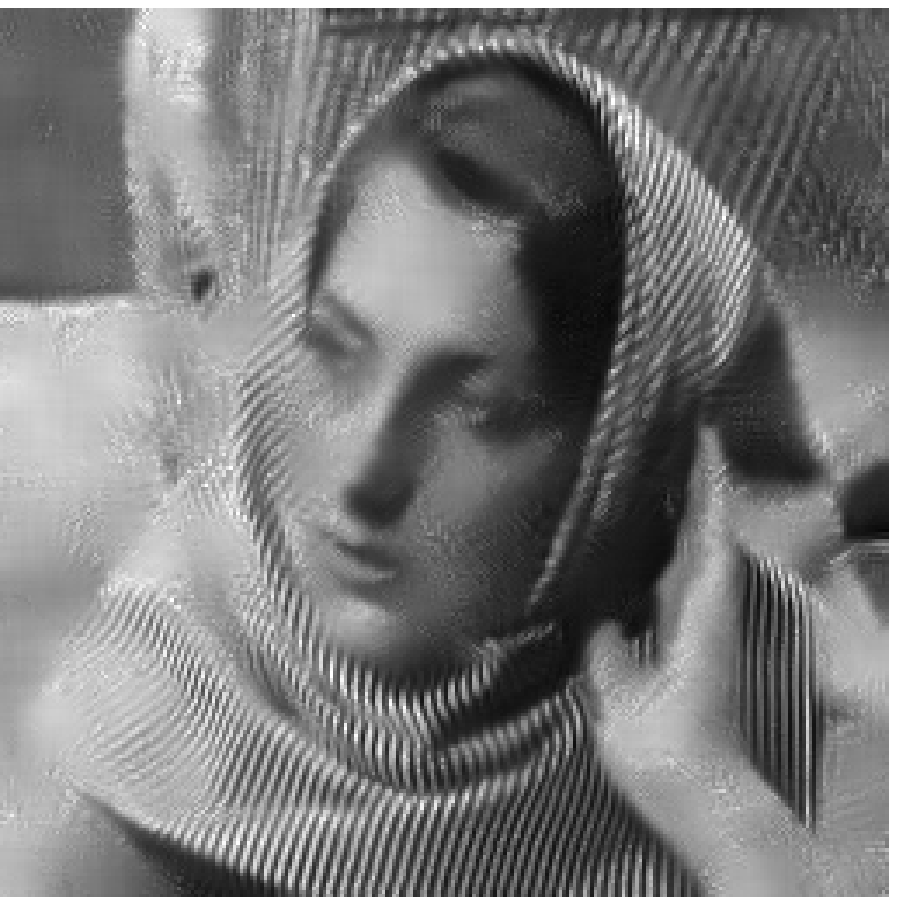} & %
\includegraphics[width=0.20\linewidth]{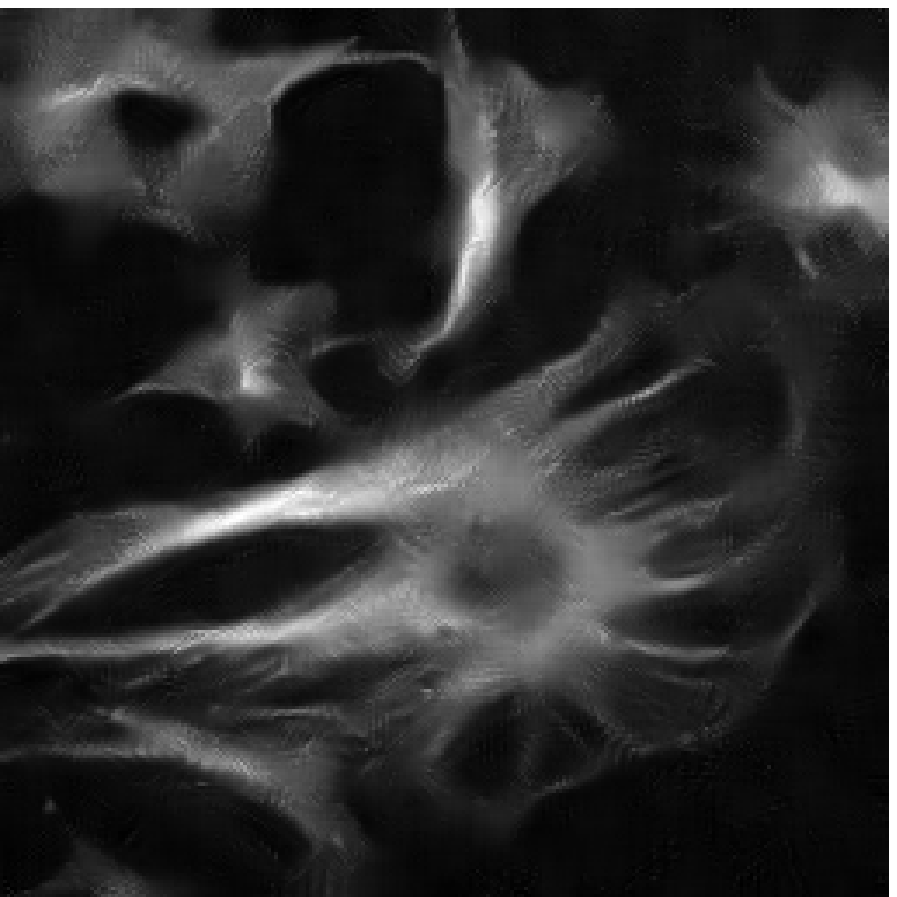} \\
NMISE=0.0739 & NMISE=0.0618 & NMISE=0.0368 & NMISE=0.1061 & NMISE=0.0855 \\
$M=19\times19$ & $M=15\times15$ & $M=9\times9$ & $M=15\times15$ & $%
M=11\times11$ \\
$m=13\times13$ & $m=5\times5$ & $m=19\times19$ & $m=21\times21$ & $%
m=17\times17$ \\
(a) Spots & (b) Galaxy & (c) Ridges & (d) Barbara & (e) Cells%
\end{tabular}
}
\caption{{\protect\small These images restored by the first step of our
algorithm. }}
\label{Fig step1}
\end{figure}
\end{center}

Note that the presented algorithm is divided into two steps: in the first
step we reconstruct the image by OWPNF from noisy data; in the second step, we
smooth the image by a Gaussian kernel. This is explained by the fact that
images with brightness between $0$ and $255$ (like Barbara) are well
denoised by the first step, but for the low count levels images, the
restored images by OWPNF are not smooth enough (see Figure \ref{Fig step1}).
For these types of images, we introduce an additional smoothing using a
Gaussian kernel (see the second step of the algorithm).

Our numerical experiments are done in the same way as in %
\citep{ZHANG2008WAVELETS} and \citep{MAKITALO2009INVERSION} to produce
comparable results; we also use the same set of test images (all of $%
256\times 256$ in size): Spots $[0.08,4.99]$, Galaxy $[0,5]$, Ridges $%
[0.05,0.85]$, Barbara $[0.93,15.73]$, and Cells $[0.53,16.93]$. The authors
of \citep{ZHANG2008WAVELETS} and \citep{MAKITALO2009INVERSION} kindly
provided us with their programs and the test images.

Figures \ref{Fig spots}- \ref{Fig Cells} illustrate the visual quality of
the denoised images using OWPNF,EUI+BM3D  \citep{makitalo2011optimal},
 MS-VST + $7/9$ \citep{ZHANG2008WAVELETS},  MS-VST + B3 \citep{ZHANG2008WAVELETS}
  and PH-HMP \citep{lefkimmiatis2009bayesian}.

Table \ref{table1} shows the NMISE values of images reconstructed by OWPNF,
EUI+BM3D, MS-VST + $7/9$, and MS-VST + B3. For Spots $[0.08,4.99]$ and
Galaxy $[0,5]$, our results are the best; for Ridges $[0.05,0.85]$, Barbara $%
[0.93,15.73]$, and Cells $[0.53,16.93]$, the method EUI+BM3D gives the
best results, but our method is also very competitive.

\begin{center}
\begin{table}[tbp]
\begin{center}
\vskip3mm {\footnotesize
\begin{tabular}{l|rrrrr}
\hline
Algorithm & Our       & EUI    & MS-VST    & MS-VST   &PH-HMT\\
Algorithm & algorithm &  +BM3D &   + $7/9$ &  + B3    &\\\hline
Spots$[0.08,4.99]$ & $\mathbf{0.0259}$ & $0.0358$ & $0.0602$ & $0.0810$ &$0.048$\\
Galaxy$[0,5]$ & $\mathbf{0.0285}$ & $0.0297$ & $0.0357$ & $0.0338$ &$0.030$\\
Ridges$[0.05,0.85]$ & $0.0162$ & $\mathbf{0.0121}$ & $0.0193$ & $0.0416$ &$-$\\
Barbara$[0.93,15.73]$ & $0.1061$ & $\mathbf{0.0863}$ & $0.2391$ & $0.3777$&$0.159$
\\
Cells$[0.53,16.93]$ & $0.0794$ & $\mathbf{0.0643}$ & $0.0909$ & $0.1487$ &$0.082$\\
\hline
\end{tabular}
} \vskip1mm
\end{center}
\caption{Comparison between Optimal Weights Filter, MS-VST + $7/9$, and
MS-VST + B3.}
\label{table1}
\end{table}
\end{center}

\section{\label{Sec:Appendix Proofs} Proofs of the main results}

\subsection{\label{Sec: proof of Th weights 001} Proof of Theorem \protect
\ref{Th weights 001}}

We begin with some preliminary results. The following lemma is similar to
Theorem 1 of Sacks and Ylvisaker \citep{Sacks1978linear} where, however, the
inequality constraints are absent.

\begin{lemma}
\label{Lemma weights}Let $g_{\rho }(w)$ be defined by (\ref{def gw}). Then
there are unique weights $w_{\rho }$ which minimize $g_{\rho }(w)$ subject
to (\ref{s2wx}), given by
\begin{equation}
w_{\rho }(x)=\frac{1}{f(x)}(b-\lambda \rho (x))^{+},
\label{eq Lemma W000}
\end{equation}%
where $b$ and $\lambda $ are uniquely determined by the following two equations:
\begin{eqnarray}
\sum_{x\in \mathbf{U}_{x_{0},h}}\frac{1}{f(x)}(b-\lambda \rho
(x))^{+} &=&1,  \label{eq Lemma W001} \\
\sum_{x\in \mathbf{U}_{x_{0},h}}\frac{1}{f(x)}(b-\lambda \rho
(x))^{+}\rho (x) &=&\lambda .  \label{eq Lemma W002}
\end{eqnarray}
\end{lemma}

\begin{proof}
Consider the Lagrange function%
\begin{equation*}
G(w,b_{0},b)=g_{\rho }(w)-2b\left(\sum_{x\in \mathbf{U}_{x_{0},h}}w(x)-1\right)-2%
\sum_{x\in \mathbf{U}_{x_{0},h}}b(x)w(x),
\end{equation*}%
where $b_{0}\in \mathbb{R}$ and $b\in \mathbb{R}^{\mathrm{card}\, \left(\mathbf{U}%
_{x_{0},h}\right) }$ is a vector with components $b(x)\geq 0,$ $x\in \mathbf{%
U}_{x_{0},h}.$ Let $w_{\rho }$ be a minimizer of $g_{\rho }\left(w\right) $
under the constraints (\ref{s2wx}). By standard results (cf. Theorem 2.2 of
Rockafellar (1993 \citep{rockafellar1993lagrange}); see also Theorem 3.9 of Whittle (1971 \citep{Wh})), there
are Lagrange multipliers $b_{0}\in \mathbb{R}$ and $b(x)\geq 0,$ $x\in
\mathbf{U}_{x_{0},h}$ such that the following Karush-Kuhn-Tucker conditions
hold: for any $x\in \mathbf{U}_{x_{0},h},$%
\begin{equation}
\frac{\partial }{\partial w\left(x\right) }G\left(w\right) \bigg|%
_{w=w_{\rho }}=2\lambda \rho (x)+2f(x)w_{\rho }(x)-2b-2b(x)=0,
\label{s5kw1}
\end{equation}
with
\begin{equation}
\lambda =\sum_{y\in \mathbf{U}_{x_{0},h}}w_{\rho }(y)\rho (y),
\label{lambda def}
\end{equation}%
and%
\begin{eqnarray}
\frac{\partial }{\partial b_{0}}G\left(w\right) \bigg|_{w=w_{\rho }}
&=&\sum_{y\in \mathbf{U}_{x_{0},h}}w_{\rho }(y)-1=0,  \label{Lb0} \\
\frac{\partial }{\partial b\left(x\right) }G\left(w\right) \bigg|%
_{w=w_{\rho }} &=&w_{\rho }(x)\;\left\{
\begin{array}{cc}
=0, & \text{if }b\left(x\right) >0, \\
\geq 0, & \text{if }b\left(x\right) =0.%
\end{array}%
\right.   \label{Lbx}
\end{eqnarray}
(Notice that the gradients of the equality constraint function $h\left(w\right) =\sum_{x\in \mathbf{U}_{x_{0},h}}$   $w(x)-1$ and of the active
inequality constraints $h_{x}\left(w\right) =w\left(x\right) ,$ $x\in
\mathbf{U}_{x_{0},h},$ are always linearly independent, since the number of
inactive inequality constraints is strictly less than $\mathrm{card}\, ~\mathbf{U}%
_{x_{0},h}.$)

If $b\left(x\right) =0,$ then by (\ref{Lbx}) we have $w_{\rho }\left(x\right) \geq
0,$ so that by (\ref{s5kw1}) we obtain $b-\lambda \rho (x)=f(x) w_{\rho
}(x)\geq 0$ and%
\begin{equation*}
w_{\rho }(x)=\frac{(b-\lambda \rho (x))^{+}}{f(x) }.
\end{equation*}
If $b\left(x\right) >0,$ then by (\ref{Lb0}) we have $w_{\rho }\left(x\right) =0.$
Taking into account (\ref{s5kw1}) we obtain
\begin{equation}
b-\lambda \rho (x)= -b(x)\leq 0,  \label{s5bl2}
\end{equation}
so that%
\begin{equation*}
w_{\rho }(x)=0=\frac{(b-\lambda \rho (x))^{+}}{f(x) }.
\end{equation*}
As to conditions (\ref{eq Lemma W001}) and (\ref{eq Lemma W002}), they
follow immediately from the constraint (\ref{Lb0}) and the equation (\ref%
{lambda def}).

Since the system (\ref{eq Lemma W001}) and (\ref{eq Lemma W002}) has a
unique solution (this can be verified by substituting $\frac{b}{\lambda }=a$%
), the minimizer of $g_{\rho }\left(w\right) $ is also unique. The
assertion of the Lemma is proved.
\end{proof}

Now we turn to the proof of Theorem \ref{Th weights 001}. Applying Lemma \ref%
{Lemma weights} with $a=b/\lambda $, we see that the unique optimal weights $%
w_{\rho}$ minimizing $g_{\rho}(w)$ subject to (\ref{s2wx}), are given by
\begin{equation}
w_{\rho}=\frac{\lambda }{f(x)}(a-\rho (x))^{+}.
 \label{s5wl}
\end{equation}%
Since the function
\begin{equation*}
M_{\rho}(t)=\sum_{x\in \mathbf{U}_{x_{0},h}}\frac{1}{f(x)}(t-\rho (x))^{+}\rho (x)
\end{equation*}
is strictly increasing and continuous with $M_{\rho}(0)=0$ and $\lim\limits_{t%
\rightarrow \infty }M_{\rho}(a)=+\infty ,$ the equation
\begin{equation*}
M_{\rho}(a)=1
\end{equation*}%
has a unique solution on $(0,\infty)$. The equation (\ref{eq Lemma W002}) together with (\ref{s5wl})
imply  (\ref{eq th weights 001}).

\subsection{\label{Sec: proof of Th oracle 001}Proof of Theorem \protect\ref%
{Th oracle 001}}

First assume that $\rho \left(x\right) =\rho _{f,x_{0}}(x)=\left\vert f\left(x\right)
-f\left(x_{0}\right) \right\vert .$ Recall that $g_{\rho}$ and $w_{\rho}$ were defined by
($\ref{def gw}$) and ($\ref{eq th weights 001}$).
Using the H\"{o}lder condition (\ref%
{Local Holder cond}) we have, for any $w$,
\begin{equation*}
g_{\rho}(w_{\rho})\leq g_{\rho}(w)\leq \overline{g}(w),
\end{equation*}%
where
\begin{equation}
\overline{g}(w)=\left(\sum_{x\in \mathbf{U}_{x_{0},h}}w(x)L\Vert
x-x_{0}\Vert _{\infty }^{\beta }\right) ^{2}+\Gamma\sum_{x\in \mathbf{U}%
_{x_{0},h}}w^{2}(x)
\label{defi gw}
\end{equation}%
and $\Gamma$ is a constant satisfying (\ref{defi Gamma}). Denote $\overline{w}=\arg \min_{w}\overline{g}(w)$. Since $w_{\rho}$
minimize $g_{\rho}(w)$ and $\rho(x)\leq L\|x-x_0\|_{\infty}^{\beta}$,
  we get
\begin{equation*}
g_{\rho}(w_{\rho})\leq g_{\rho}(\overline{w})\leq\overline{g}(\overline{w}).
\end{equation*}%
By Theorem \ref{Th weights 001},
\begin{equation}
\overline{w}(x)=\left(\overline{a}-L\Vert x-x_{0}\Vert _{\infty }^{\beta
}\right) ^{+}\Big/\sum\limits_{x^{\prime }\in \mathbf{U}_{x_{0},h}}\left(\overline{a}-L\Vert x^{\prime }-x_{0}\Vert _{\infty }^{\beta }\right) ^{+},
\label{defi w_overline}
\end{equation}%
where $\overline{a}>0$ is the unique solution on $(0,\infty)$ of the equation $\overline{M}_h(\overline{a})=1$, where
\begin{equation}
\overline{{M}}_{h}\left(t\right) =\sum_{x\in \mathbf{U}_{x_{0},h}}\frac{1}{%
\Gamma}L\Vert x-x_{0}\Vert _{\infty }^{\beta }(t-L\Vert x-x_{0}\Vert
_{\infty }^{\beta })^{+}>0 .
\label{defi Mt}
\end{equation}%
Now Theorem \ref{Th oracle 001} is a consequence of the following lemma.

\begin{lemma}
\label{lm5_2}Assume that $\rho(x) = L\|x-x_0\|_{\infty}^{\beta}$ and that $h\geq c_{0}n^{-\alpha }$ with $0\leq \alpha <\frac{%
1}{2\beta +2}$, or $h=c_{0}n^{-\frac{1}{2\beta +2}}$ with $c_{0}\geq
c_{1}\left(L,\beta \right) =\left(\frac{(2\beta +2)(\beta +2)\Gamma}{8L^{2}%
}\right) ^{\frac{1}{2\beta +2}}.$ Then
\begin{equation}
\overline{a}=c_{3}n^{-\beta /(2\beta +2)}(1+o(1))  \label{s5hk}
\end{equation}%
and
\begin{equation}
\overline{g}(\overline{w})\leq c_{4}n^{-\frac{2\beta }{2+2\beta }}(1+o(1)),
\label{s5gw4}
\end{equation}%
where $c_{3}$ and $c_{4}$ are constants depending only on $\beta $ and $L$.
\end{lemma}

\begin{proof}
We first prove (\ref{s5hk}) in the  case where $h=1$ i.e. $\mathbf{U}_{x_0,h}=\mathbf{I}$.
 In this case by the
definition of $\overline{a},$ we have
\begin{equation}
\overline{M}_{1}\left(\overline{a}\right) =\sum\limits_{x\in \mathbf{I}}\frac{1}{\Gamma}(%
\overline{a}-L\Vert x-x_{0}\Vert _{\infty }^{\beta })^{+}L\Vert x-x_{0}\Vert
_{\infty }^{\beta }=1.  \label{eq - M1 001}
\end{equation}%
Let $\overline{h}=\left(\overline{a}/L\right) ^{1/\beta }$.
 Then $a- L\|x-x_0\|_{\infty}^2\geq 0$ if and only if $\|x-x_0\|\leq \overline{h}$. So from
 (\ref{eq - M1 001}) we get
\begin{equation}
L^{2}\overline{h}^{\beta }\sum_{\Vert x-x_{0}\Vert _{\infty }\leq \overline{h%
}}\Vert x-x_{0}\Vert _{\infty }^{\beta }-L^{2}\sum_{\Vert x-x_{0}\Vert
_{\infty }\leq \overline{h}}\Vert x-x_{0}\Vert _{\infty }^{2\beta }=\Gamma.
\label{eq - M1 002}
\end{equation}%
By the definition of the neighborhood $\mathbf{U}_{x_{0},\overline{h}}$,
it is easily seen that%
\begin{equation*}
\sum_{\Vert x-x_{0}\Vert _{\infty }\leq \overline{h}}\Vert x-x_{0}\Vert
_{\infty }^{\beta }=8N^{-\beta }\sum_{k=1}^{N\overline{h}}k^{\beta +1}=8N^{2}%
\frac{\overline{h}^{\beta +2}}{\beta +2}\left(1+o\left(1\right) \right)
\end{equation*}%
and%
\begin{equation*}
\sum_{\Vert x-x_{0}\Vert _{\infty }\leq \overline{h}}\Vert x-x_{0}\Vert
_{\infty }^{2\beta }=8N^{-2\beta }\sum_{k=1}^{N\overline{h}}k^{2\beta
+1}=8N^{2}\frac{\overline{h}^{2\beta +2}}{2\beta +2}\left(1+o\left(1\right) \right).
\end{equation*}%
Therefore, (\ref{eq - M1 002}) implies
\begin{equation*}
\frac{8L^{2}\beta }{\left(\beta +2\right) \left(2\beta +2\right) }N^{2}%
\overline{h}^{2\beta +2}(1+o{(1)})=\Gamma,
\end{equation*}%
from which we infer that%
\begin{equation}
\overline{h}=c_{1}n^{-\frac{1}{2\beta +2}}(1+o(1))  \label{eq - h bar}
\end{equation}%
with $c_{1}=\left(\Gamma\frac{\left(\beta +2\right) \left(2\beta
+2\right) }{8L^{2}\beta }\right) ^{\frac{1}{2\beta +2}}.$ From ($\ref{eq - h bar}$)
and the definition of $\overline{h}$, we obtain
\begin{equation*}
\overline{a}=L\overline{h}^{\beta }=Lc_{1}^{\beta }n^{-\frac{%
\beta }{2\beta +2}}(1+o(1)),
\end{equation*}
which proves (\ref{s5hk}) in the case where $h=1$.

We next prove (\ref{eq - h bar}), which implies(\ref{s5hk}), under the conditions of the lemma.
 First, notice that if $\overline{h}\leq h \leq 1$, then $\overline{M}_h(t)=\overline{M}_1(t)$, $\forall t >0$.
 If $h\geq c_{0}n^{-\alpha },$ where $0\leq \alpha <\frac{1}{2\beta +2},$
then it is clear that $h\geq \overline{h},$ for $n$ sufficiently large.
Therefore $\overline{M}_{h}\left(\overline{a}\right) =\overline{M}_{1}\left(\overline{a}\right) $%
, thus we arrive at the equation (\ref{eq - M1 001}), from which we deduce (%
\ref{eq - h bar}). If $h\geq c_{0}n^{-\frac{1}{2\beta +2}}$ and $c_{0}\geq
c_{1}, $ then again $h\geq \overline{h}$ for $n$ sufficiently large.
Therefore, $\overline{M}_{h}\left(\overline{a}\right) =\overline{M}_{1}\left(\overline{a}\right)
$, and we arrive again at (\ref{eq - h bar}).

We finally prove  (\ref{s5gw4}). Denote for brevity
\begin{equation*}
G_{h}=\sum_{\Vert x-x_{0}\Vert _{\infty }\leq h}(\overline{h}^{\beta }-\Vert
x-x_{0}\Vert _{\infty }^{\beta })^{+}.
\end{equation*}%
Since $h\geq \overline{h},$ for $n$ sufficiently large, we have $\overline{{M}}%
_{h}\left(\overline{a}\right) =\overline{{M}}_{\overline{h}}\left(\overline{a%
}\right) =1$ and $G_{h}=G_{\overline{h}}.$ Therefore by (\ref{defi gw}), (\ref{defi w_overline})
and (\ref{defi Mt}), we get
\begin{eqnarray*}
\overline{g}(\overline{w}) &=&{\Gamma}\frac{\overline{M}_{\overline{h}%
}\left(\overline{a}\right) +\sum_{\Vert x-x_{0}\Vert _{\infty }\leq
\overline{h}}\left(\left(\overline{a}-L\Vert x-x_{0}\Vert _{\infty
}^{\beta }\right) ^{+}\right) ^{2}}{L^{2}G_{\overline{h}}^{2}} =\frac{{\Gamma}}{L}\frac{\overline{a}}{G_{\overline{h}}}.
\end{eqnarray*}%
Since
\begin{align*}
G_{h}& =\sum_{\Vert x-x_{0}\Vert _{\infty }\leq \overline{h}}(\overline{h}%
^{\beta }-\Vert x-x_{0}\Vert _{\infty }^{\beta }) \\
& =\overline{h}^{\beta }\sum_{1\leq k\leq N\overline{h}}8k-\frac{8}{N^{\beta
}}\sum_{1\leq k\leq N\overline{h}}k^{\beta +1} \\
& =\frac{4\beta }{\beta +2}N^{2}\overline{h}^{\beta +2}\left(1+o\left(1\right) \right) \\
& =\frac{4\beta }{\left(\beta +2\right) L^{1/\beta }}N^{2}\overline{a}%
^{\left(\beta +2\right) /\beta }\left(1+o\left(1\right) \right) ,
\end{align*}
we obtain
\begin{equation*}
\overline{g}\left(\overline{w}\right) =\Gamma\frac{\left(\beta +2\right) }{%
4\beta }L^{1/\beta -1}\frac{\overline{a}^{-\frac{2}{\beta }}}{N^{2}}\left(1+o\left(1\right) \right) \leq c_{4}n^{-\frac{2\beta }{2\beta +2}}\left(1+o\left(1\right) \right) ,
\end{equation*}%
where $c_{4}$ is a constant depending on $\beta $ and $L$.
\end{proof}

{\bf{Proof of Theorem \ref{Th oracle 001}}}. As $\rho
\left( x\right) =\left\vert f\left( x\right) -f\left( x_{0}\right)
\right\vert +\delta _{n},$
we have
\begin{equation*}
\begin{split}
\left( \sum_{x\in \mathbf{U}_{x_0,h}} w(x)\rho(x)\right)^2
&\leq
\left( \sum_{x\in \mathbf{U}_{x_0,h}} w(x) |f(x)-f(x_0)|+\delta_n\right)^2
\\& \leq
2\left(\sum_{x\in \mathbf{U}_{x_0,h}} w(x) |f(x)-f(x_0)|  \right)^2+2\delta_n^2.
\end{split}
\end{equation*}
Since $f(x)\leq \Gamma,$ with $g_{\rho}$ and $\overline{g}$ by (\ref{def gw})
and (\ref{defi gw}), we get
\begin{equation*}
g_{\rho}(w)\leq 2 \overline{g}(w) + 2\delta_n^2.
\end{equation*}
So
\begin{equation*}
g_{\rho}(w_{\rho})\leq g_{\rho}(\overline{w})\leq 2\overline{g}(\overline{w})+2\delta
_{n}^{2}.
\end{equation*}
Therefore, by  Lemma \ref{lm5_2} and the condition that $\delta_n=
O\left( n^{-\frac{\beta}{2\beta+2}}\right)$, we obtain
\begin{equation*}
g_{\rho}(w_{\rho})=O\left(n^{-\frac{2\beta }{2\beta +2}}\right).
\end{equation*}%
This together with (\ref{s2ef-vers2})  give (\ref{s2ef2}).

\subsection{\label{Sec: proof of Th adapt 001}Proof of Theorem \protect\ref%
{Th adapt 001}}

Let $M'=\mathrm{card}\, \mathbf{U}_{x_{0},h}^{\prime }=(2Nh+1)^{2}/2$, $m'=\mathrm{card}\, \mathbf{U}''_{x_{0},\eta
}=(2N\eta+1) ^{2}/2$. Denote $\Delta _{x_{0},x}\left(y\right) =f(y)-f(Ty)$ and $\eta
\left(y\right) =\epsilon (y)-\epsilon (Ty)$,  where $\epsilon$ is defined by (\ref{defi epsilon}).
It is easy to see that%
\begin{equation*}
\frac{1}{m'}\sum_{y\in \mathbf{U}''_{x_{0},\eta }}\left(Y(y)-Y(Ty)\right) ^{2}=%
\frac{1}{m'}\sum_{y\in \mathbf{U}''_{x_{0},\eta }}\Delta ^{2}_{x_{0},x}\left(y\right) +\frac{1}{m'}S(x)+\overline{f}(x_{0})+\overline{f}(x),
\end{equation*}%
where $S(x)=S_{2}(x)-S_{1}(x),$ with
\begin{eqnarray*}
S_{1}(x) &=&2\sum_{y\in \mathbf{U}''_{x_{0},\eta }}\Delta _{x_{0},x}\left(y\right) \eta \left(y\right) , \\
S_{2}(x) &=&\sum_{y\in \mathbf{U}''_{x_{0},\eta }}\left(\eta^{2} \left(y\right)
-\overline{f}(x_{0})-\overline{f}(x)\right) .
\end{eqnarray*}%
Denote%
\begin{equation*}
V=\frac{1}{m'}\sum_{y\in \mathbf{U}''_{x_{0},\eta }}\Delta ^{2} _{x_{0},x}\left(y\right)+\frac{1}{m'}S(x).
\end{equation*}%
Then
\begin{equation}
\begin{split}
\widehat{\rho }_{x_0}(x)& =\left(\sqrt{V+\overline{f}(x_{0})+\overline{f}(x)}-%
\sqrt{2\overline{f}(x_{0})}\right) ^{+} \\
& \leq \left\vert \sqrt{V+\overline{f}(x_{0})+\overline{f}(x)}-\sqrt{2%
\overline{f}(x_{0})}\right\vert .
\end{split}%
\label{ineq rho}
\end{equation}%
Using one-term Taylor expansion, we obtain
\begin{equation}
\begin{split}
& \left\vert \sqrt{V+\overline{f}(x_{0})+\overline{f}(x)}-\sqrt{2\overline{f}%
(x_{0})}\right\vert \\
& \leq \frac{|V|}{(\overline{f}(x_{0})+\overline{f}(x)+\theta V)^{1/2}}%
+\left\vert \sqrt{\overline{f}(x_{0})+\overline{f}(x)}-\sqrt{2\overline{f}%
(x_{0})}\right\vert \\
& \leq \frac{\frac{1}{m'}\sum_{y\in \mathbf{U}''_{x_{0},\eta }}\Delta^{2}_{x_0,x} (y)+%
\frac{1}{m'}|S(x)|}{(\overline{f}(x_{0})+\overline{f}(x)+\theta V)^{1/2}}%
+\left\vert \frac{\overline{f}(x)-\overline{f}(x_{0})}{ \sqrt{%
\overline{f}(x_{0})+\overline{f}(x)}+\sqrt{2\overline{f}(x_{0})} }%
\right\vert .
\end{split}%
\label{ineq vff}
\end{equation}%
Since $f(x)\geq 1/\ln n$, $x\in \mathbf{I}$, (\ref{ineq rho}) and (\ref{ineq vff})
imply that
\begin{equation}
\widehat{\rho }_{x_0}(x)\leq \frac{L^{2}h^{2\beta }+\frac{1}{m'}|S(x)|}{(2/\ln
n+\theta V)^{1/2}}+c_{3}h\sqrt{\ln n}.  \label{eq rho hate}
\end{equation}

We shall use three lemmas to finish the Proof of Theorem \ref{Th adapt 001}.

The following lemma can be deduced form the results in Borovkov %
\citep{borovkov2000estimates}, see also Merlevede, Peligrad and Rio \citep{merlevède2010bernstein}.

\begin{lemma}
\label{Lemma Borovk} If, for some $\delta >0,\gamma \in (0,1)$ and $K>1$ we
have
\begin{equation*}
\sup \mathbb{E}\exp \left(\delta \left\vert X_{i}\right\vert ^{\gamma }\right) \leq
K,\;i=1,...,n,
\end{equation*}%
then there are two positive constants $c_{1}$ and $c_{2}$ depending only on $%
\delta ,$ $\gamma $ and $K$ such that, for any $t>0,$
\begin{equation*}
\mathbb{P}\left(\sum_{i=1}^{n}X_{i}\geq t\right) \leq \exp \left(-c_{1}t^{2}/n\right) +n\exp \left(-c_{2}t^{\gamma }\right) .
\end{equation*}
\end{lemma}

\begin{lemma}
\label{lm5_4} Assume that $h=c_{0}n^{-\frac{1}{2\beta +2}%
}$ with $c_{0}>c_{1}=\left(\Gamma\frac{\left(\beta +2\right) \left(2\beta
+2\right) }{8L^{2}\beta }\right) ^{\frac{1}{2\beta +2}}$ and that $\eta
=c_{2}n^{-\frac{1}{2\beta +2}}.$ Suppose that the function $f$ satisfies the local H\"{o}lder
condition (\ref{Local Holder cond}).  Then there exists a constant $c_{4}>0$
depending only on $\beta $ and $L$, such that
\begin{equation}
\mathbb P\left\{ \max_{x\in \mathbf{U}'_{x_{0},h}} \widehat{\rho }_{x_0}(x)\geq c_{4}n^{-%
\frac{\beta}{2\beta+2}}\sqrt{\ln n}\right\} = O\left(n^{-\frac{2\beta}{%
2\beta+2}}\right).  \label{s5pr}
\end{equation}
\end{lemma}

\begin{proof}
Note that
\begin{equation*}
\mathbb{E}e^{Y(y)}=\sum_{k=0}^{+\infty }e^{k}\frac{f^{k}(x)e^{-f(x)}}{k!}%
=ef(y)e^{(e-1)f(y)}\leq e\Gamma e^{(e-1)\Gamma }.
\end{equation*}%
From this inequality we easily deduce that
\begin{eqnarray*}
\sup_{y} \mathbb{E}\left(e^{|Z(y)|^{1/2}}\right)
&\leq&
 \sup_{y} \left(\mathbb{E}e^{Y(T_{x,x_{0}}y)+Y(y)+\sqrt{2(\Gamma ^{2}+\Gamma)}}\right)
 \\&\leq&
  \left(e\Gamma \right) ^{2}e^{2(e-1)\Gamma +\sqrt{2(\Gamma ^{2}+\Gamma)}},
\end{eqnarray*}%
where
\begin{equation*}
Z(y)=2\Delta _{x_{0},x}\left(y\right) \eta \left(y\right) +\left(\eta^{2}
\left(y\right) -\overline{f}(x_{0})-\overline{f}(x)\right) .
\end{equation*}%
By Lemma \ref{Lemma Borovk}, we infer that there exists two positive
constants $c_{5}$ and $c_{6}$ such that
\begin{equation}
\mathbb{P}\left(\frac{1}{m'}|S(x)|\geq z/\sqrt{m'}\right)\leq \exp (-c_{5}z^{2})+m'\exp (-c_{6}(%
\sqrt{m'}z)^{\frac{1}{2}}).  \label{large number theorem}
\end{equation}%
Substituting $z=\sqrt{\frac{1}{c_{5}}\ln m'M'}$ into the inequality (\ref%
{large number theorem}), we see that for $m'$ large enough,
\begin{equation*}
\mathbb{P}\left(\frac{1}{m'}\left\vert S(x)\right\vert \geq \frac{\sqrt{\frac{1}{c_{5}%
}\ln m'M'}}{\sqrt{m'}}\right) \leq 2\exp \left(-\ln m'M'\right) =\frac{2}{m'M'}.
\end{equation*}%
From this inequality we easily deduce that
\begin{eqnarray*}
&&\mathbb{P}\left(\max_{x\in \mathbf{U}_{x_{0},h}^{\prime }}\frac{1}{m'}\left\vert
S(x)\right\vert \geq \frac{\sqrt{\frac{1}{c_{5}}\ln m'M'}}{\sqrt{m'}}\right)
\\ &
\leq & \sum_{x\in \mathbf{U}_{x_{0},h}^{\prime }}
\mathbb{P}\left(\frac{1}{m'}\left\vert
S(x)\right\vert \geq \frac{\sqrt{\frac{1}{c_{5}}\ln m'M'}}{\sqrt{m'}}\right)
\leq \frac{2}{m'}.
\end{eqnarray*}%
Taking $M'=(2Nh+1)^{2}/2=c_{0}^{2}n^{\frac{2\beta }{2\beta +2}}/2$ and $m'=(2N\eta+1)
^{2}/2=c_{2}^{2}n^{\frac{2\beta }{2\beta +2}}/2$, we arrive at
\begin{equation}
\mathbb{P}\left(\mathbf{B}\right) \leq c_{7}n^{-\frac{2\beta }{2\beta +2}},
\label{ineq p b}
\end{equation}%
where $\mathbf{B}=\{\max_{x\in \mathbf{U}_{x_{0},h}^{\prime }}\frac{1}{m'}%
\left\vert S(x)\right\vert \geq c_{8}n^{-\frac{\beta }{2\beta +2}}\sqrt{\ln n%
}\}$ and $c_{8}$ is a constant depending only on $\beta $ and $L$. Since on
the set $\mathbf{B}$ we have
\begin{equation}
\left(\frac{2}{\ln n}+\theta V\right) ^{1/2}<\frac{1}{\sqrt{\ln n}}
\label{ineq sigma}
\end{equation}%
for $n$ large enough, combining (\ref{eq rho hate}), (\ref{ineq p b}) and (%
\ref{ineq sigma}), we get (\ref{s5pr}).
\end{proof}


\begin{lemma}
\label{s5ef} Suppose that the conditions of Theorem \ref{Th adapt 001} are
satisfied. Then
\begin{equation*}
\mathbb{P}\left(\mathbb{E}\{|\widehat{f}_{h}(x_{0})-f(x_{0})|^{2}\big|Y(x),x\in \mathbf{I}''_{x_{0}}\}\geq c_{9}n^{-\frac{2\beta }{2\beta +2}}\ln n\right)=O(n^{-\frac{%
2\beta}{2\beta+2}}),
\end{equation*}%
where $c_{9}>0$ is a constant depending only on $\beta$ and $L$. \label%
{lm5_6}
\end{lemma}


\begin{proof}
Taking into account (\ref{s3rx}), (\ref{s3fh}) and the independence of $%
\epsilon (x)$, we have
\begin{equation}
\begin{split}
& \mathbb{E}\{|\widehat{f}_{h}(x_{0})-f(x_{0})|^{2}\big|Y(x),x\in \mathbf{I}''_{x_{0}}\} \\
& \leq \left(\sum_{x\in \mathbf{U}_{x_{0},h(x)}^{\prime }}\widehat{w}%
(x)\rho (x)\right) ^{2}+\overline{f}(x)\sum_{x\in \mathbf{U}%
_{x_{0},h(x)}^{\prime }}\widehat{w}^{2}(x).
\end{split}
\label{eq Lemma experance}
\end{equation}%
Since $\rho(x)<Lh^{\beta} $, from (\ref{eq Lemma experance}) we get
\begin{equation*}
\begin{split}
& \mathbb{E}\{|\widehat{f}_{h}(x_{0})-f(x_{0})|^{2}\big|Y(x),x\in \mathbf{I}''_{x_{0}}\} \\
& \leq \left(\sum_{x\in \mathbf{U}_{x_{0},h}^{\prime }}\widehat{w}%
(x)Lh^{\beta} \right) ^{2}+\overline{f}(x)\sum_{x\in \mathbf{U}%
_{x_{0},h}^{\prime }}\widehat{w}^{2}(x) \\
& \leq \left(\left(\sum_{x\in \mathbf{U}_{x_{0},h}^{\prime }}\widehat{w}(x)%
\widehat{\rho }_{x_0}(x)\right) ^{2}+\overline{f}(x)\sum_{x\in \mathbf{U}%
_{x_{0},h}^{\prime }}\widehat{w}^{2}(x)\right) +L^2h^{2\beta}.
\end{split}%
\end{equation*}%
Recall that $\widehat{w}(x)$ stand for the optimal weights, defined by (\ref{s3ww}). Therefore
\begin{equation}
\begin{split}
&\mathbb{ E}\{|\widehat{f}_{h}(x_{0})-f(x_{0})|^{2}\big|Y(x),x\in \mathbf{I}''_{x_{0}}\} \\
& \leq \left(\left(\sum_{x\in \mathbf{U}_{x_{0},h}^{\prime }}{\overline{w}}%
_{1}(x)\widehat{\rho }_{x_0}(x)\right) ^{2}+\overline{f}(x)\sum_{x\in \mathbf{U}%
_{x_{0},h}^{\prime }}{\overline{w}}_{1}^{ 2}(x)\right)+L^2h^{2\beta},
\end{split}
\label{s5ef4}
\end{equation}%
where $\overline{w}_{1}=\arg \min\limits_{w}\overline{g}_{1}(w)$ with
\begin{equation*}
\overline{g}_{1}(w)= \left(\sum_{x\in \mathbf{U}_{x_{0},h}^{\prime
}}{w}(x)L\|x-x_0\|_{\infty}^{\beta}\right) ^{2}+\Gamma\sum_{x\in \mathbf{U}_{x_{0},h}^{\prime
}}{w}^{2}(x).
\end{equation*}%
Since by Lemma \ref{lm5_4},
\begin{equation*}
\mathbb{P}\left\{ \max_{x\in \mathbf{U}_{x_{0},h}} \widehat{\rho }_{x_0}(x)< c_{4}n^{-\frac{%
\beta}{2\beta+2}}\sqrt{\ln n}\right\} = 1-O(n^{-\frac{2\beta}{2\beta+2}}),
\end{equation*}
the inequality (\ref{s5ef4}) becomes
\begin{equation*}
\begin{split}
&\mathbb{P}\left(\mathbb{E}\{|\widehat{f}_{h}(x_{0})-f(x_{0})|^{2}\big|Y(x),x\in \mathbf{I}''_{x_{0}}\} < \overline{g}_{1}(\overline{w}_{1})+2c^2_{4}n^{-\frac{2\beta}{%
2\beta+2}}\ln n+L^2h^{2\beta}\right) \\
&=1-O(n^{-\frac{2\beta}{2\beta+2}}).
\end{split}
\label{eq p e}
\end{equation*}%
Now, the assertion of the theorem is obtained easily if we note that $%
h^{2\beta}=c_{10}^{2\beta }n^{-\frac{2\beta}{2\beta+2}}$ and $\overline{g}%
_{1}(\overline{w}_{1})\leq c_{11}n^{-\frac{2\beta }{2\beta +2}}, $ for some
constant $c_{12}$ depending only on $\beta$ and $L$ (by Lemma \ref{lm5_2}
with $\mathbf{U}_{x_{0},h}^{\prime }$ instead of $\mathbf{U}_{x_{0},h}$).
\end{proof}

{\bf Proof of Theorem \ref{Th adapt 001}}. Using (\ref{eq Lemma experance}), the condition (\ref{Local Holder cond}) and bound $f(x) \leq \Gamma$ we obtain
\begin{equation*}
\mathbb{E}\left(|\widehat{f}_{h}(x_{0})-f(x_{0})|^{2}\big|Y(x),x\in \mathbf{I}''_{x_{0}},\right) \leq \overline{g}_{1}(\widehat{w})\leq c_{12},
\end{equation*}%
for a constant $c_{14}>0$ depending only on $\beta $, $L$ and $\Gamma$. Applying Lemma %
\ref{s5ef}, we have
\begin{equation*}
\begin{split}
\mathbb{E}& \left(|\widehat{f}_{h}(x_{0})-f(x_{0})|^{2}\big|Y(x),x\in \mathbf{I}''_{x_{0}},\right) \\
<& \mathbb{P}\left(E\{|\widehat{f}_{h}(x_{0})-f(x_{0})|^{2}\big|Y(x),x\in \mathbf{I}''_{x_{0}}\}<c_{9}n^{-\frac{2\beta }{2\beta +2}}\ln n\right) c_{9}n^{-\frac{%
2\beta }{2\beta +2}}\ln n \\
& +\mathbb{P}\left(E\{|\widehat{f}_{h}(x_{0})-f(x_{0})|^{2}\big|Y(x),x\in \mathbf{I}''_{x_{0}}\}\geq c_{9}n^{-\frac{2\beta }{2\beta +2}}\ln n\right) c_{12} \\
=& O\left(n^{-\frac{2\beta }{2\beta +2}}\ln n\right) ,
\end{split}%
\end{equation*}%
%
%
%
where the constant in $O$ depending only on $\beta $, $L$ and $\Gamma$. Taking
expectation proves Theorem \ref{Th adapt 001}.


\begin{center}
\begin{figure}[tbp]
\renewcommand{\arraystretch}{0.5} \addtolength{\tabcolsep}{-5pt} \vskip3mm {%
\fontsize{8pt}{\baselineskip}\selectfont
\begin{center}
\begin{tabular}{ccccc}
\includegraphics[width=0.25\linewidth]{Spots_original.eps} & %
\includegraphics[width=0.25\linewidth]{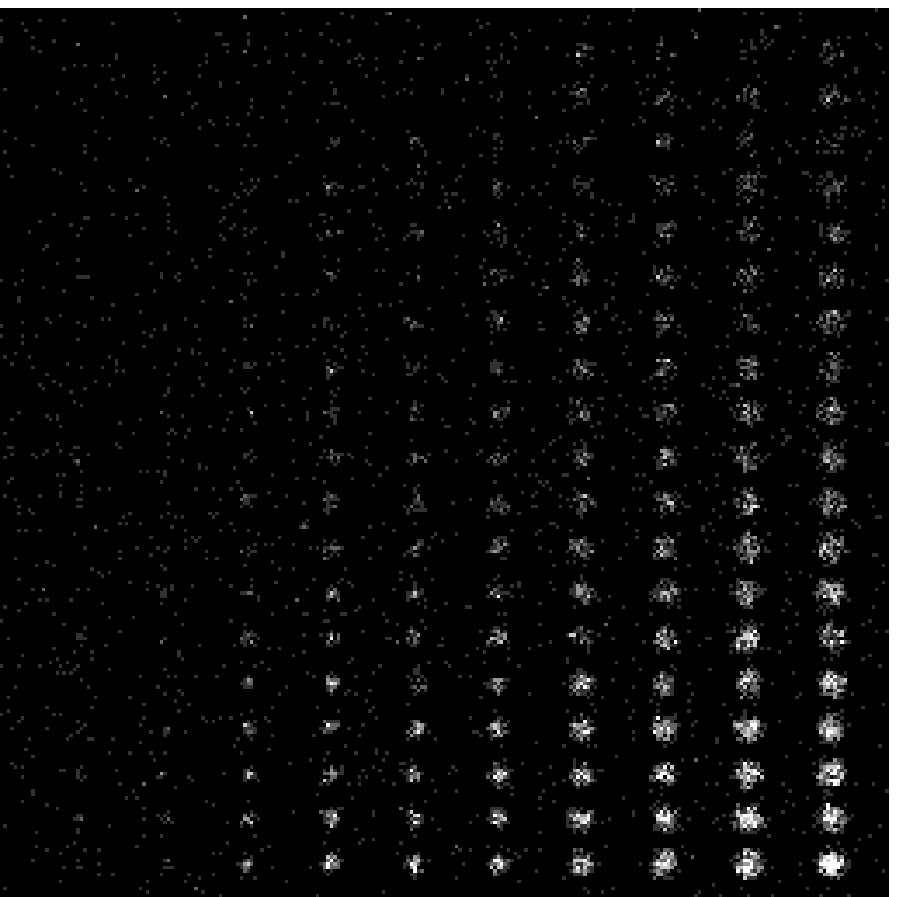} & %
\includegraphics[width=0.25\linewidth]{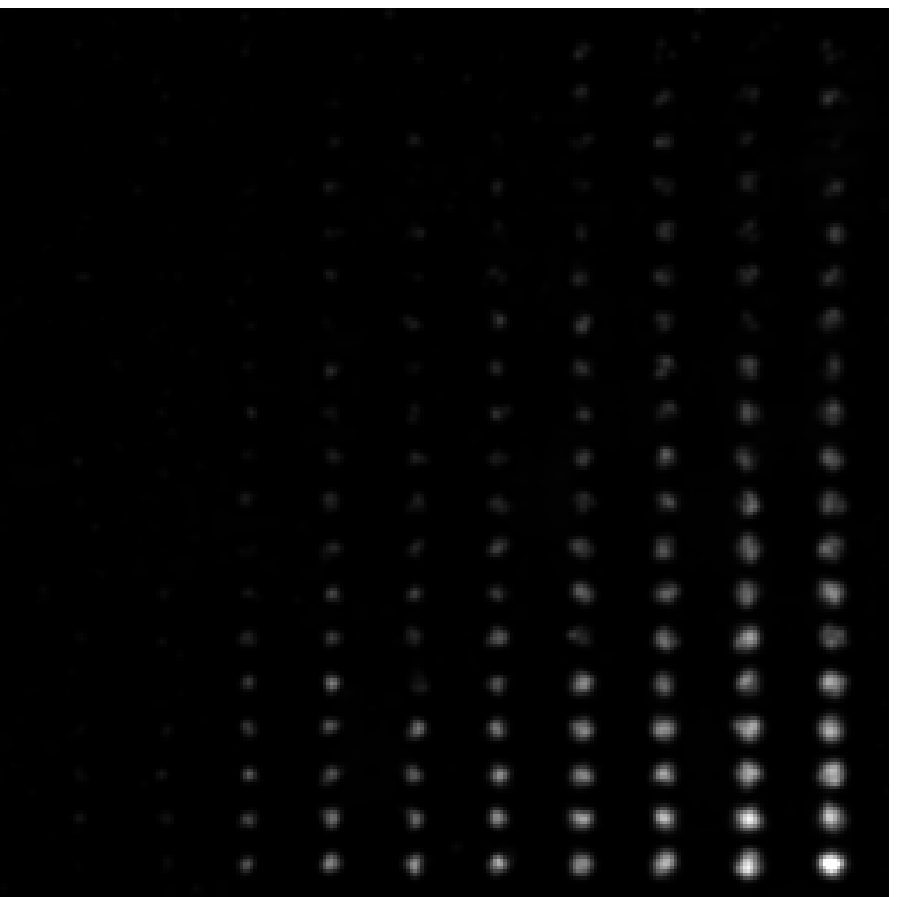} &  &  \\
(a) Original image & (b) Noisy image & (c)  OWF &  &  \\
\includegraphics[width=0.25\linewidth]{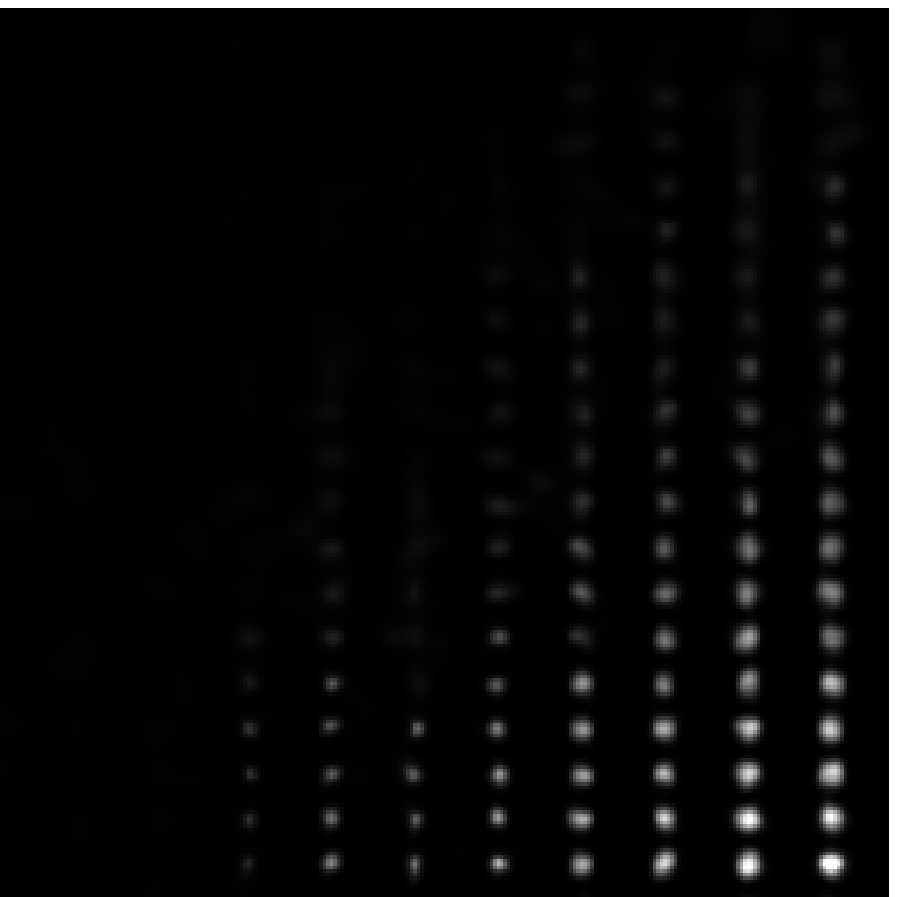} & %
\includegraphics[width=0.25\linewidth]{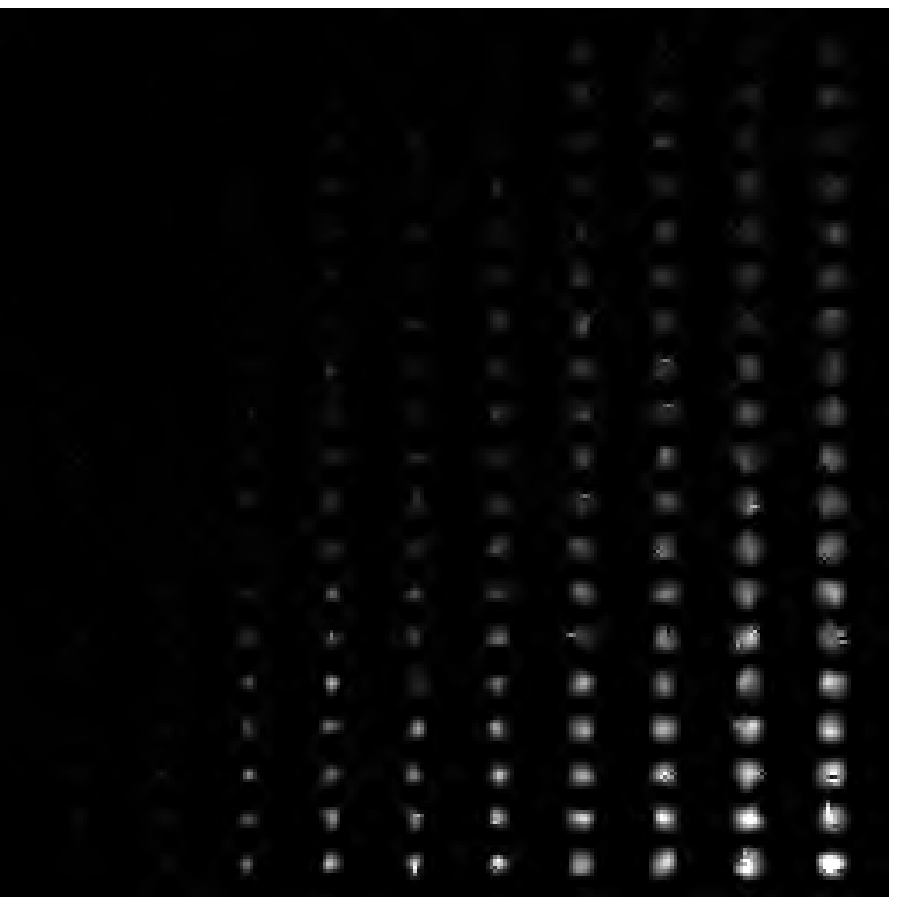} & %
\includegraphics[width=0.25\linewidth]{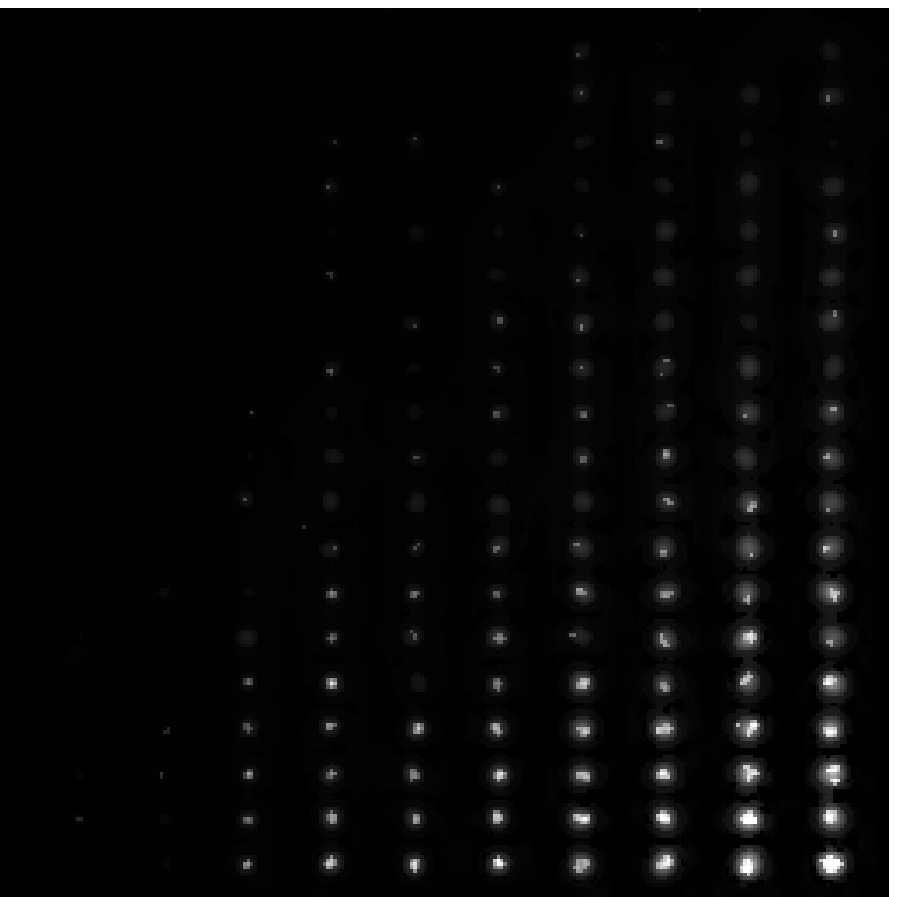} &  &  \\
(d)  EUI+BM3D & (e)  MS-VST + $7/9$ & (f)  MS-VST + B3 &  &
\end{tabular}
\end{center}
}
\caption{{\protect\small Denoising an image of simulated spots of different
radii (image size: $256 \times 256$). (a) simulated sources (amplitudes $\in
[0.08, 4.99]$; background $= 0.03$); (b) observed counts; (c) Optimal
Weights Filter ($M=19\times19$, $m=13\times13$, $d=2$ and $H=1, $ $NMISE =
0.0259$); (d) Exact unbiased inverse + BM3D ($NMISE=0.0358$) (e) MS-VST + $%
7/9$ biorthogonal wavelet ($J = 5$, $FPR = 0.01$,$N_{max}= 5$ iterations, $%
NMISE = 0.0602$); (f) MS-VST + B3 isotropic wavelet ($J = 5$, $FPR = 0.01$, $%
N_{max}= 5$ iterations, $NMISE = 0.81$).}}
\label{Fig spots}
\end{figure}
\end{center}

\begin{center}
\begin{figure}[tbp]
\renewcommand{\arraystretch}{0.5} \addtolength{\tabcolsep}{-5pt} \vskip3mm {%
\fontsize{8pt}{\baselineskip}\selectfont
\begin{center}
\begin{tabular}{ccccc}
\includegraphics[width=0.25\linewidth]{Galaxy_original.eps} & %
\includegraphics[width=0.25\linewidth]{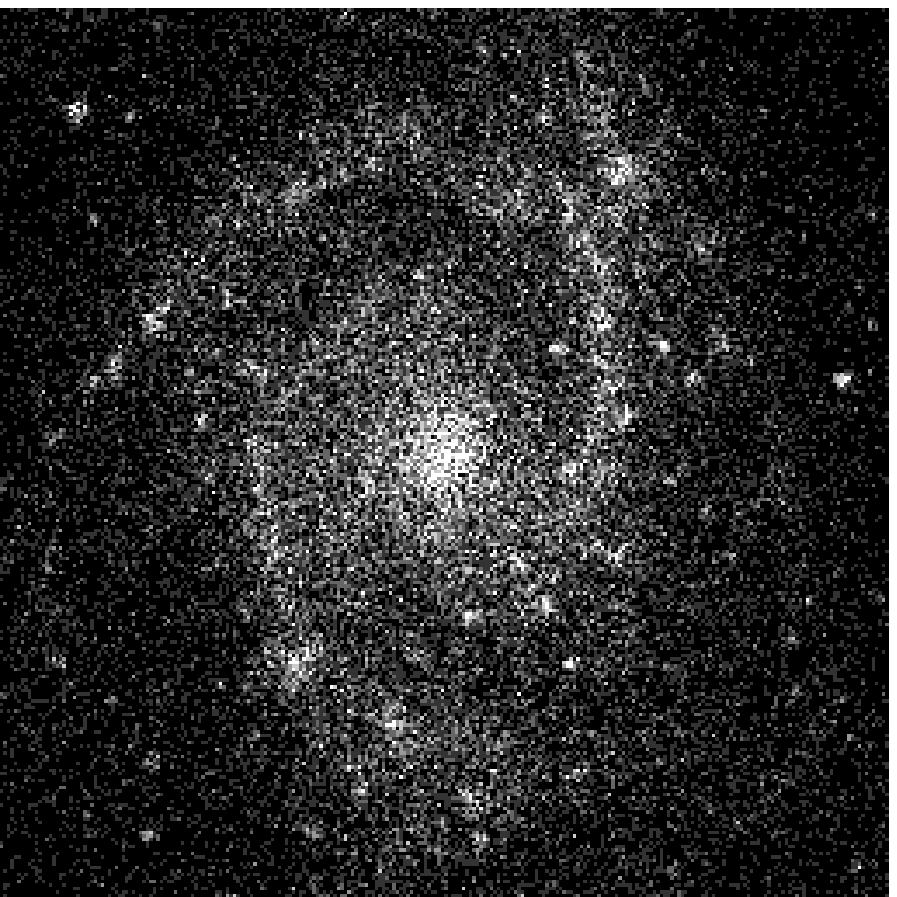} & %
\includegraphics[width=0.25\linewidth]{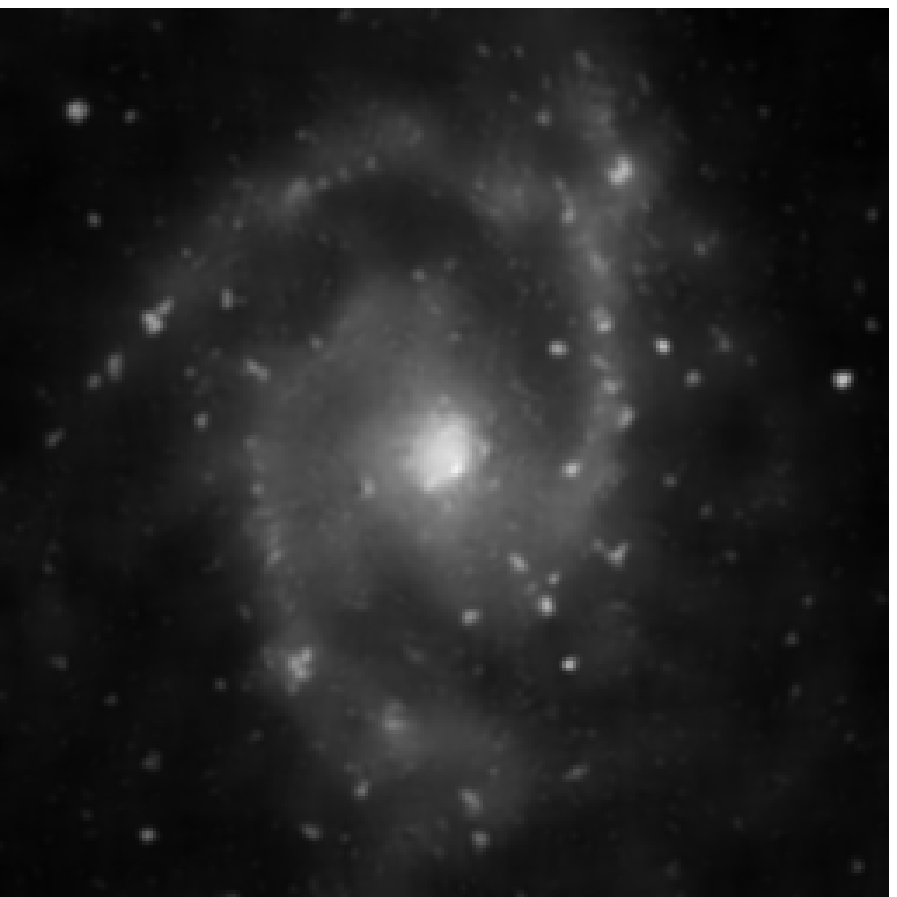} &  &  \\
(a) Original image & (b) Noisy image & (c)  OWF &  &  \\
\includegraphics[width=0.25\linewidth]{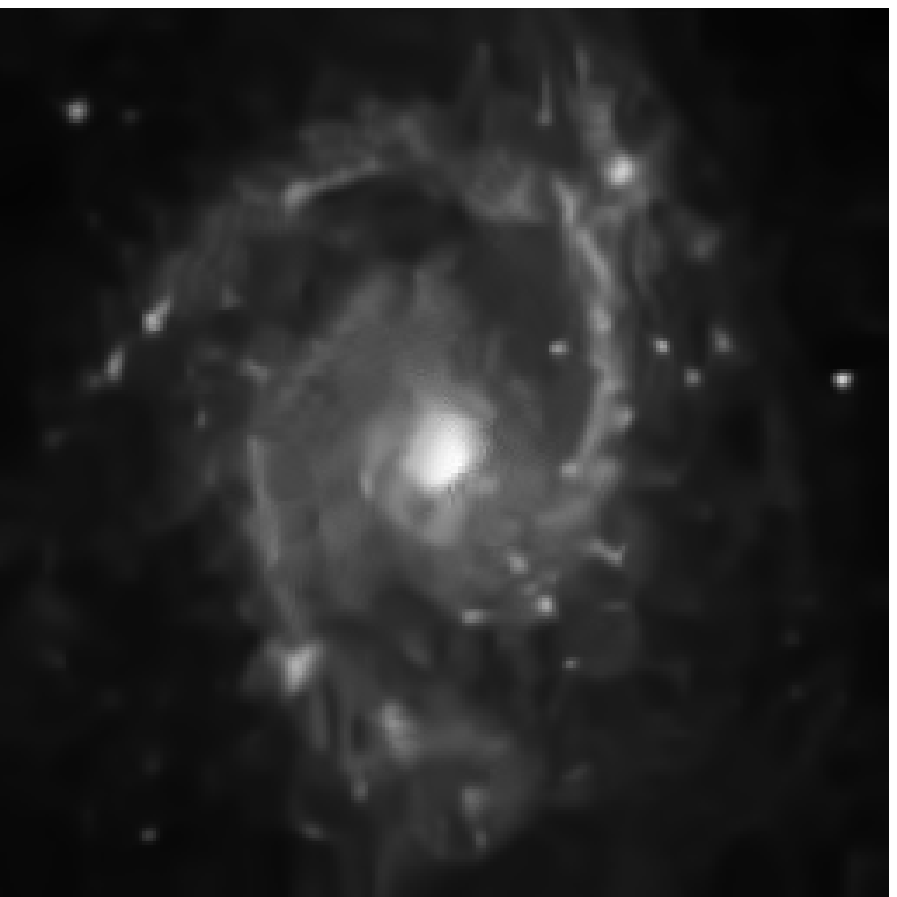} & %
\includegraphics[width=0.25\linewidth]{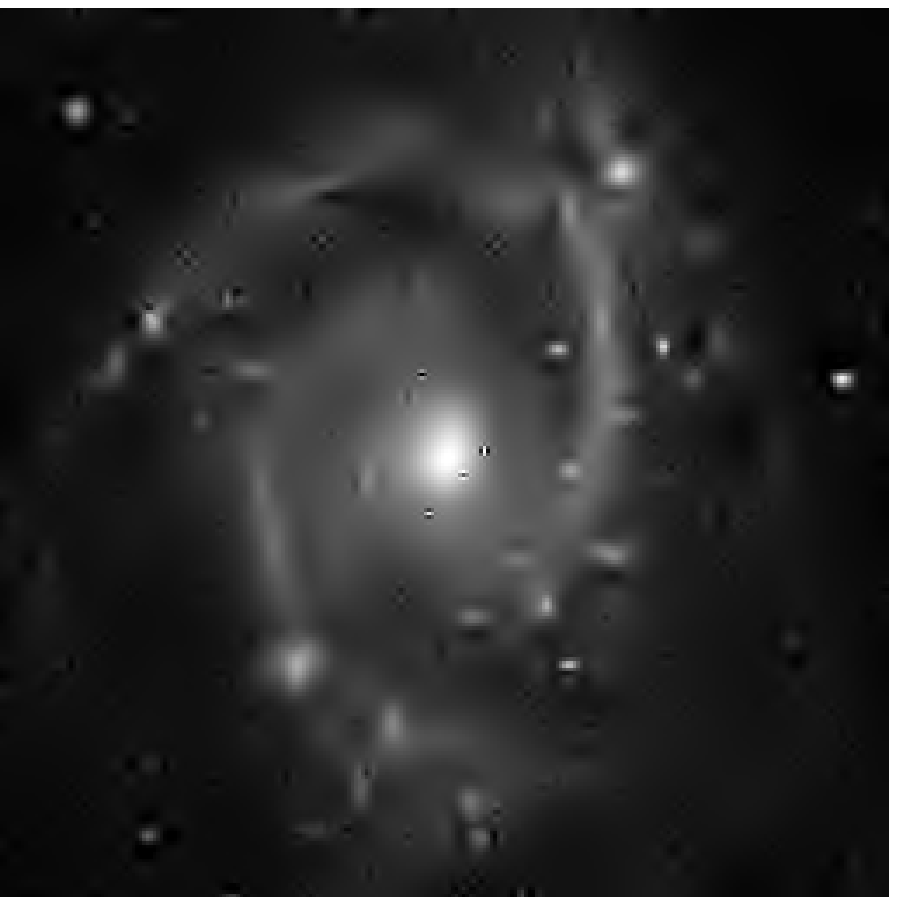} & %
\includegraphics[width=0.25\linewidth]{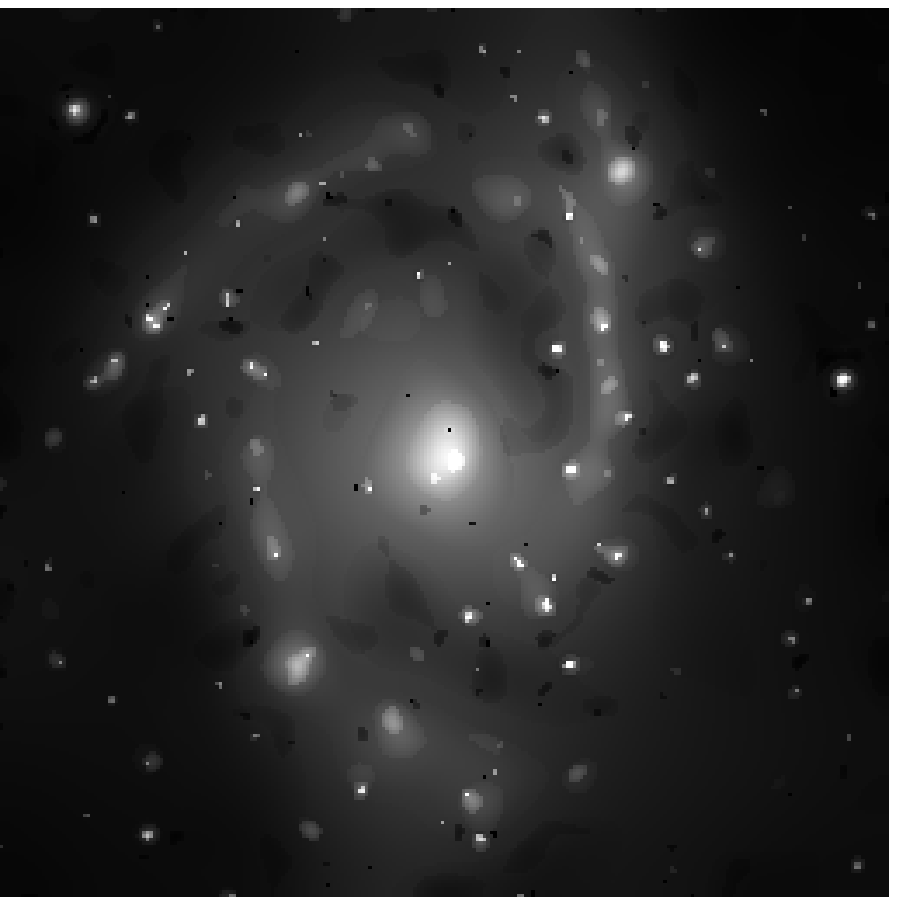} &  &  \\
(d)  EUI+BM3D & (e)  MS-VST + $7/9$ & (f)  MS-VST + B3 &  &
\end{tabular}
\end{center}
}
\caption{{\protect\small Denoising a galaxy image (image size: $256 \times
256$). (a) galaxy image (intensity $\in [0, 5]$); (b) observed counts; (c)
Optimal Weights Filter ($M=15\times15$, $m=5\times5$, $d=2$ and $H=1$, $%
NMISE = 0.0285$); (d) Exact unbiased inverse + BM3D ($NMISE=0.0297$) (e)
MS-VST + $7/9$ biorthogonal wavelet ($J = 5$, $FPR = 0.0001$, $N_{max}= 5$
iterations, $NMISE = 0.0357$); (f) MS-VST + B3 isotropic wavelet ($J = 3$, $%
FPR = 0.0001$, $N_{max}= 10$ iterations, $NMISE = 0.0338$).}}
\label{Fig Galaxy}
\end{figure}
\end{center}

\begin{center}
\begin{figure}[tbp]
\renewcommand{\arraystretch}{0.5} \addtolength{\tabcolsep}{-5pt} \vskip3mm {%
\fontsize{8pt}{\baselineskip}\selectfont
\begin{center}
\begin{tabular}{ccccc}
\includegraphics[width=0.25\linewidth]{Ridges_original.eps} & %
\includegraphics[width=0.25\linewidth]{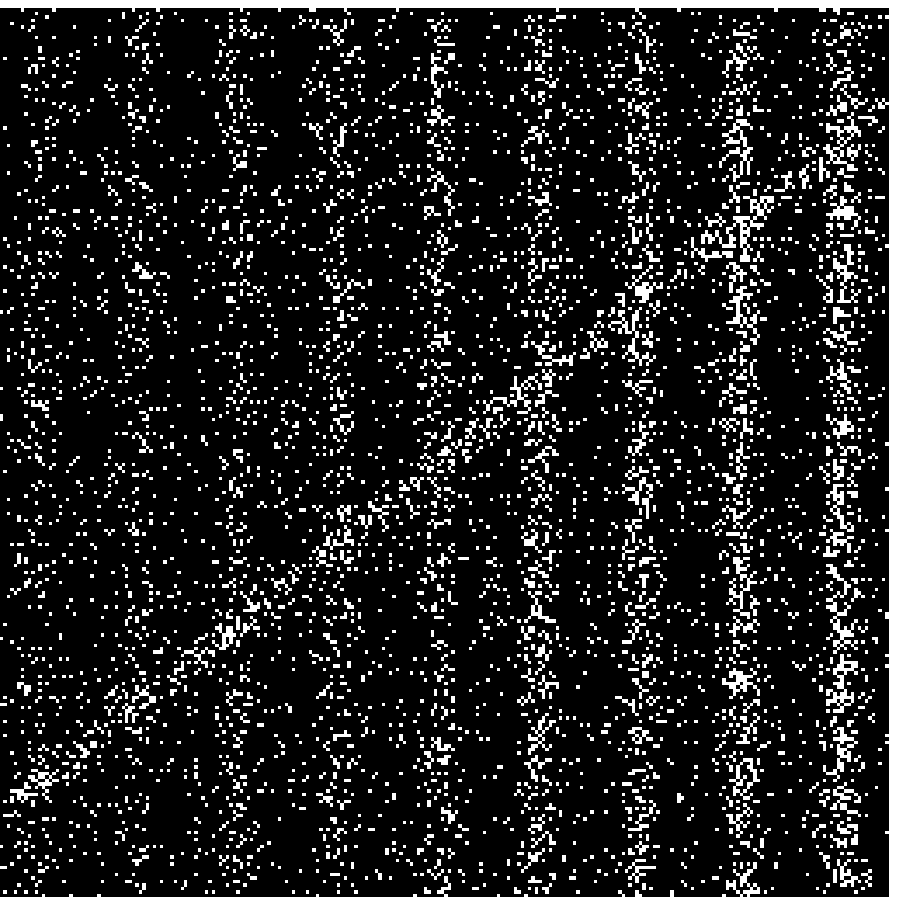} & %
\includegraphics[width=0.25\linewidth]{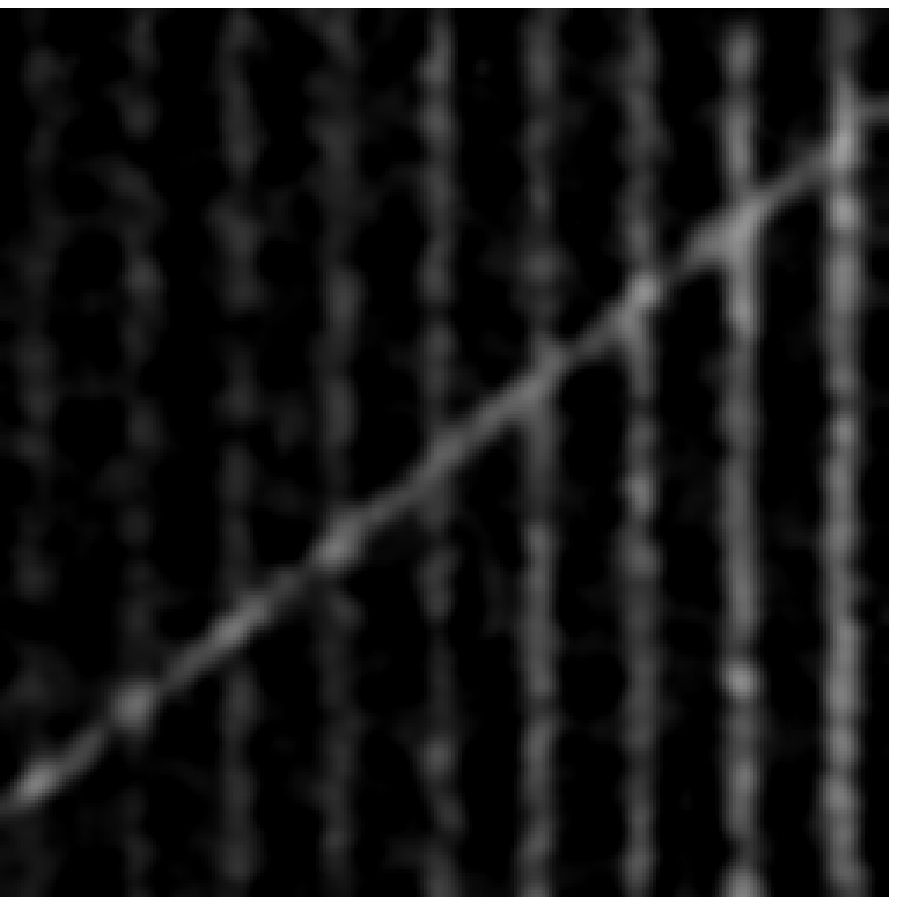} &  &  \\
(a) Original image & (b) Noisy image & (c)  OWF &  &  \\
\includegraphics[width=0.25\linewidth]{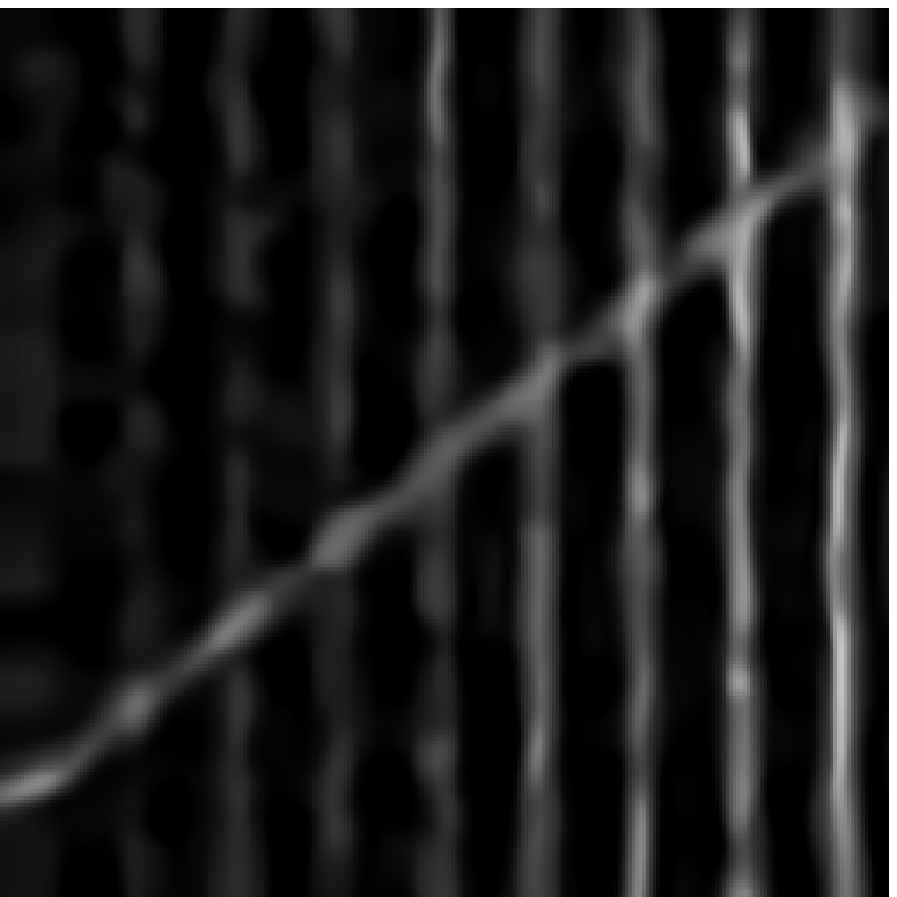} & %
\includegraphics[width=0.25\linewidth]{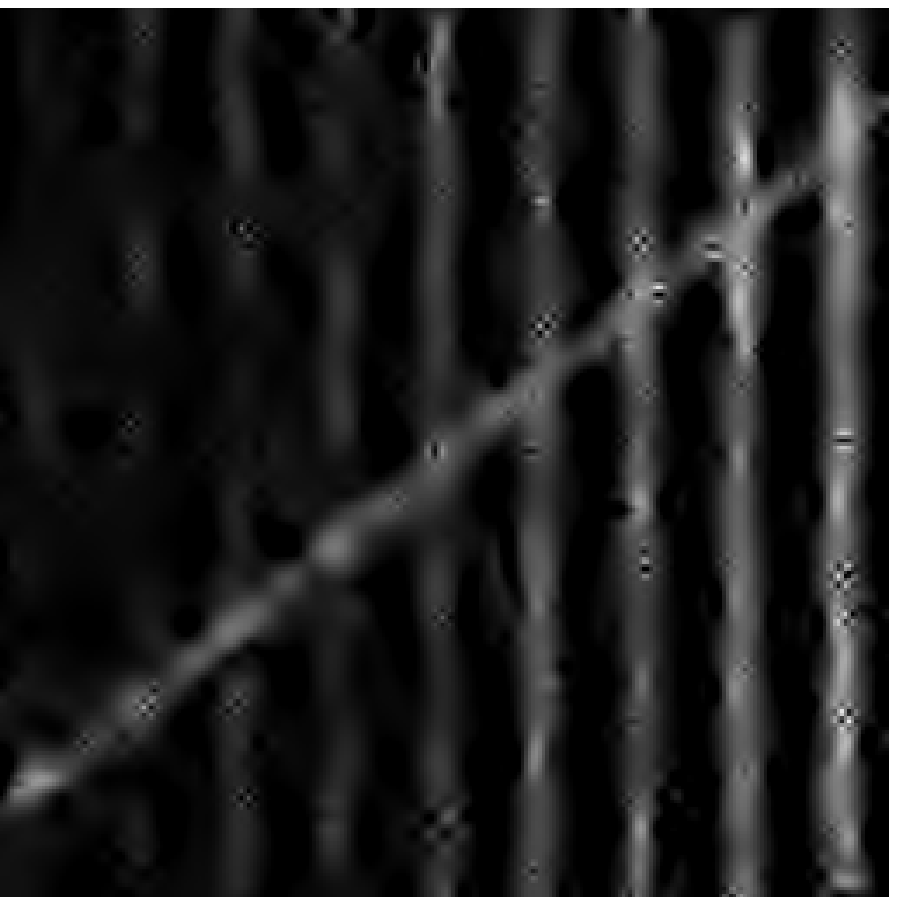} & %
\includegraphics[width=0.25\linewidth]{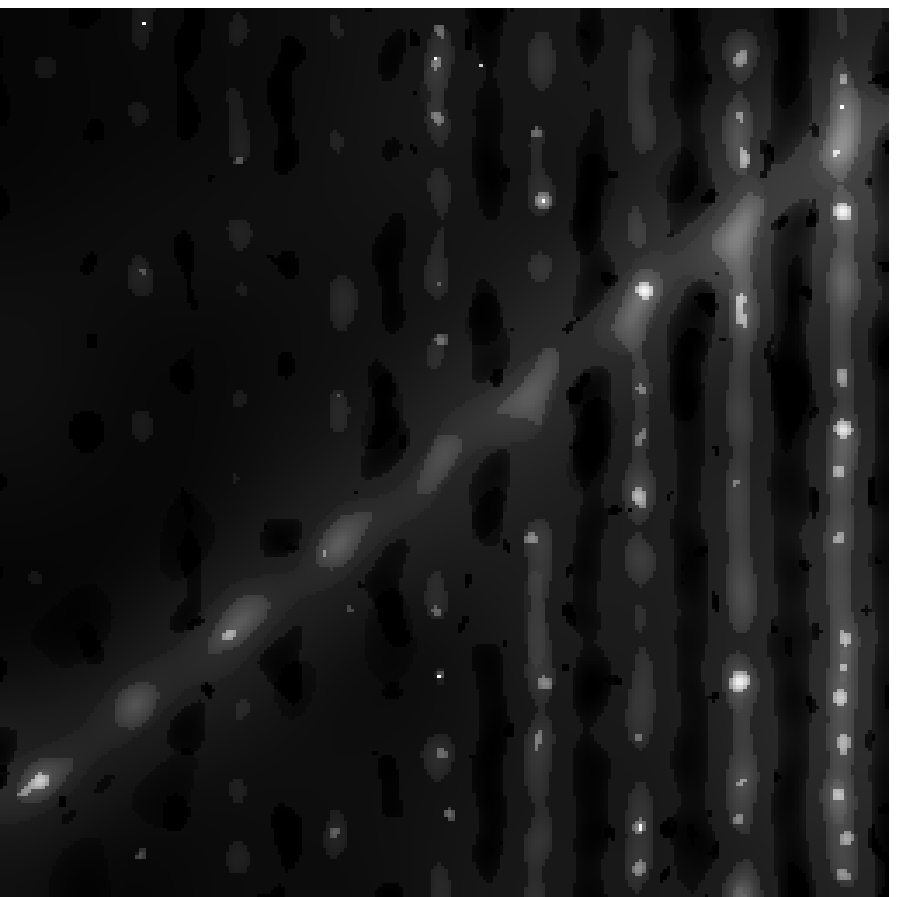} &  &  \\
(d)  EUI+BM3D & (e)  MS-VST + $7/9$ & (f)  MS-VST + B3 &  &
\end{tabular}
\end{center}
}
\caption{{\protect\small Poisson denoising of smooth ridges (image size: $%
256 \times 256$). (a) intensity image (the peak intensities of the $9$
vertical ridges vary progressively from $0.1$ to $0.5$; the inclined ridge
has a maximum intensity of $0.3$; background $= 0.05$); (b) Poisson noisy
image; (c) Optimal Weights Filter ($M=9\times9$, $m=19\times19$, $d=3$ and $%
H=2$, $NMISE = 0.0162$); (d) Exact unbiased inverse + BM3D ($NMISE=0.0121$);
 (e) MS-VST + $7/9$ biorthogonal wavelet ($J = 5$, $FPR =
0.001$, $N_{max}= 5$ iterations, $NMISE = 0.0193$); (f) MS-VST + B3
isotropic wavelet ($J = 3$, $FPR = 0.00001$, $N_{max}= 10$ iterations, $%
NMISE = 0.0416$).}}
\label{Fig Ridges}
\end{figure}
\end{center}

\begin{center}
\begin{figure}[tbp]
\renewcommand{\arraystretch}{0.5} \addtolength{\tabcolsep}{-5pt} \vskip3mm {%
\fontsize{8pt}{\baselineskip}\selectfont
\begin{center}
\begin{tabular}{ccccc}
\includegraphics[width=0.25\linewidth]{Barbara_original.eps} & %
\includegraphics[width=0.25\linewidth]{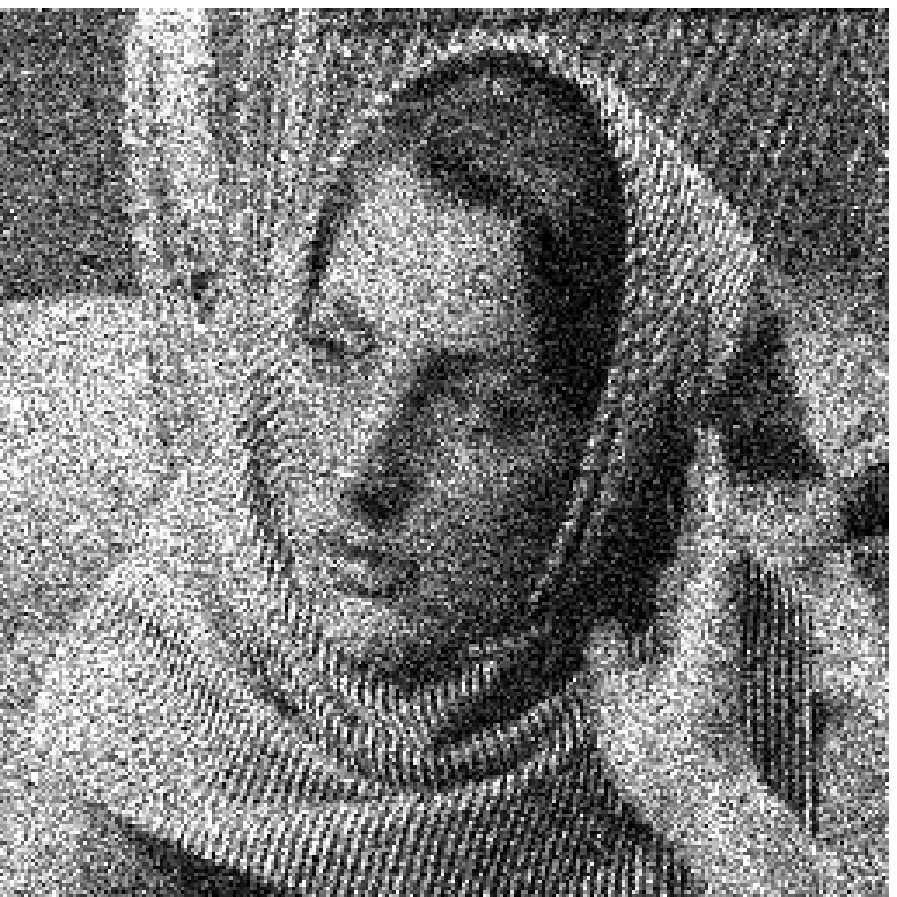} & %
\includegraphics[width=0.25\linewidth]{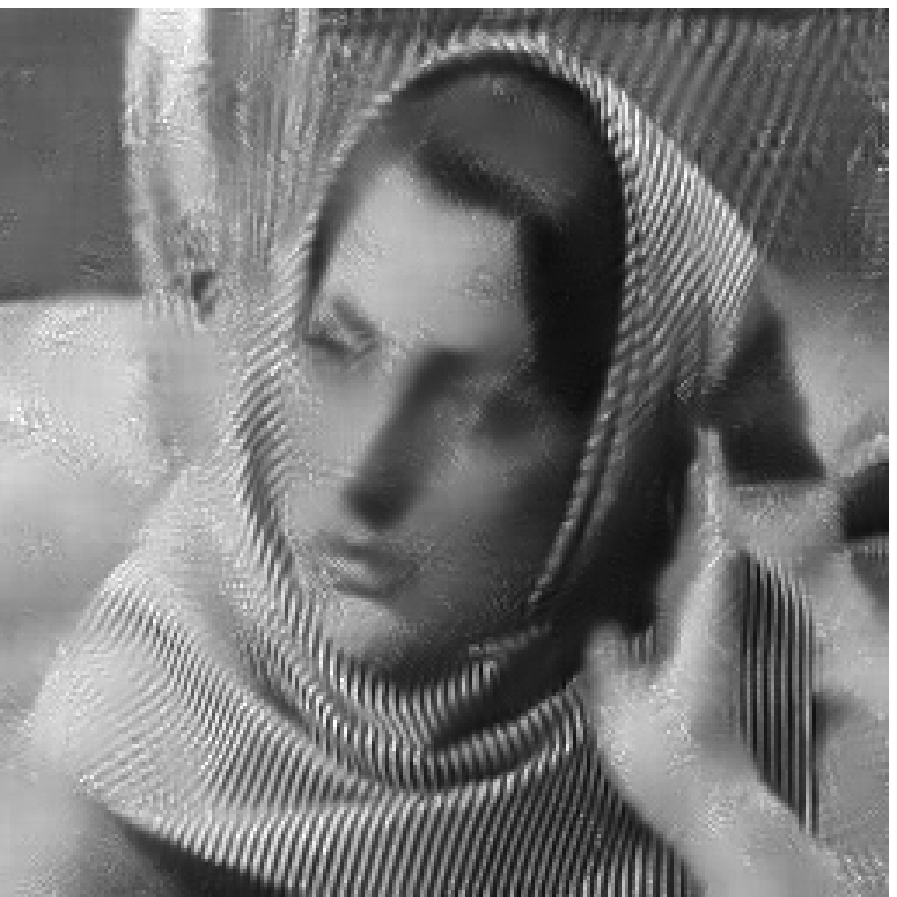} &  &  \\
(a) Original image & (b) Noisy image & (c)  OWF &  &  \\
\includegraphics[width=0.25\linewidth]{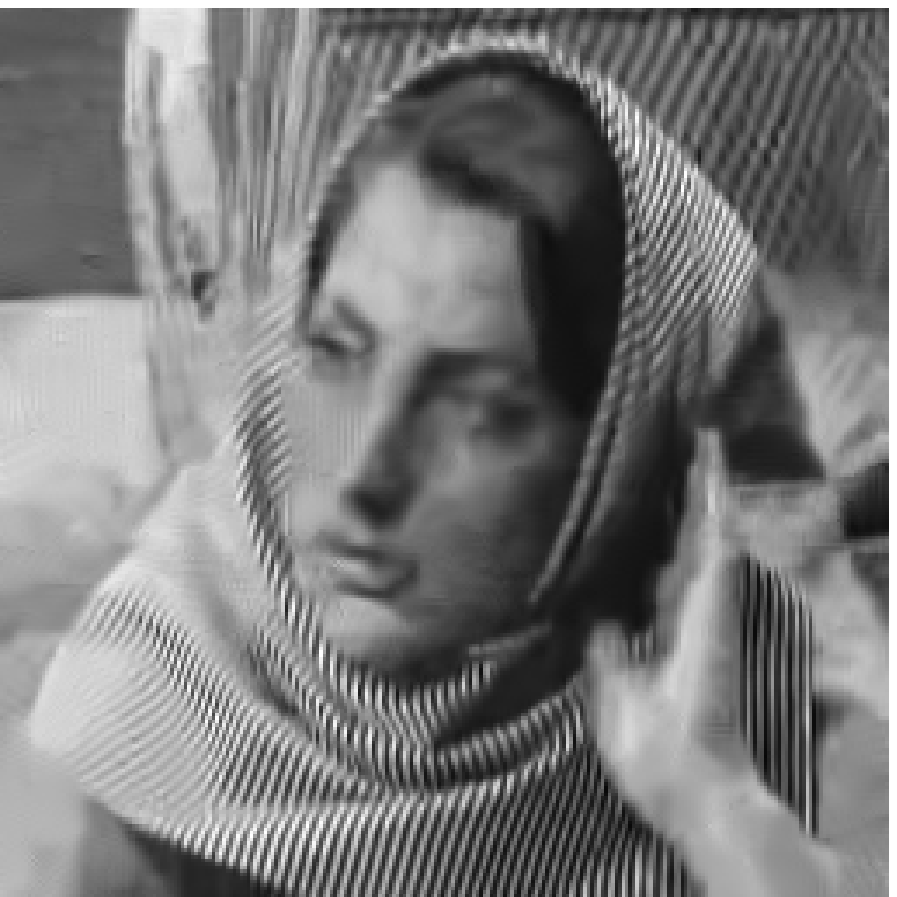} & %
\includegraphics[width=0.25\linewidth]{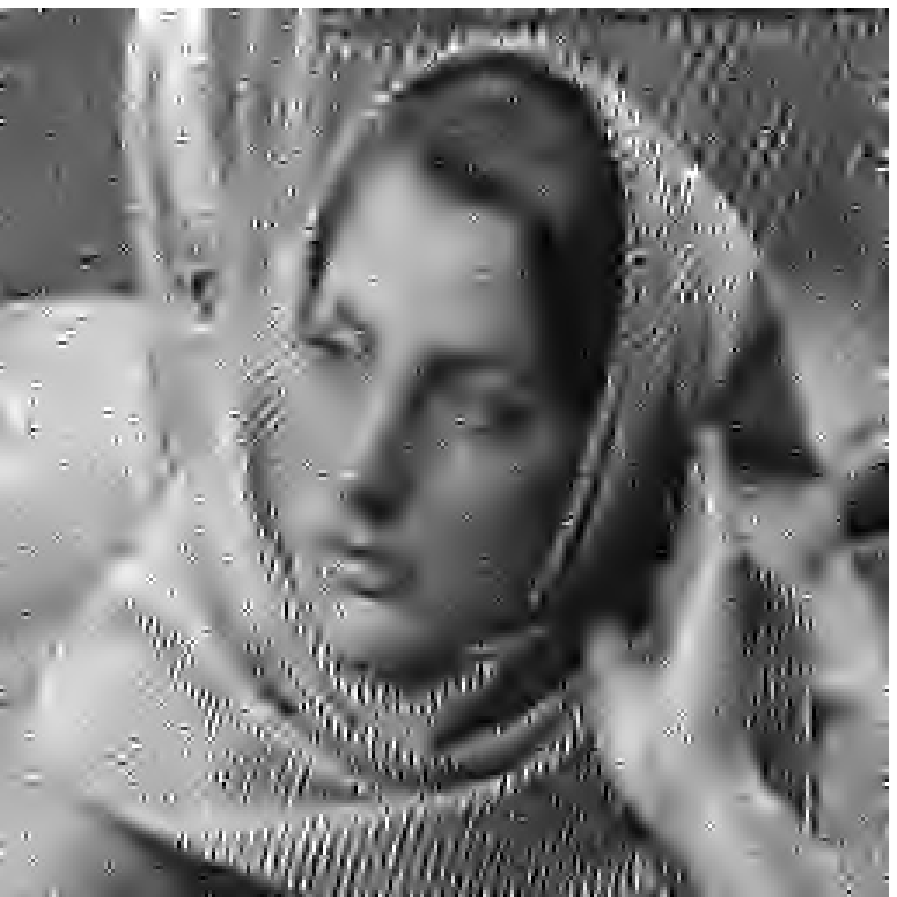} & %
\includegraphics[width=0.25\linewidth]{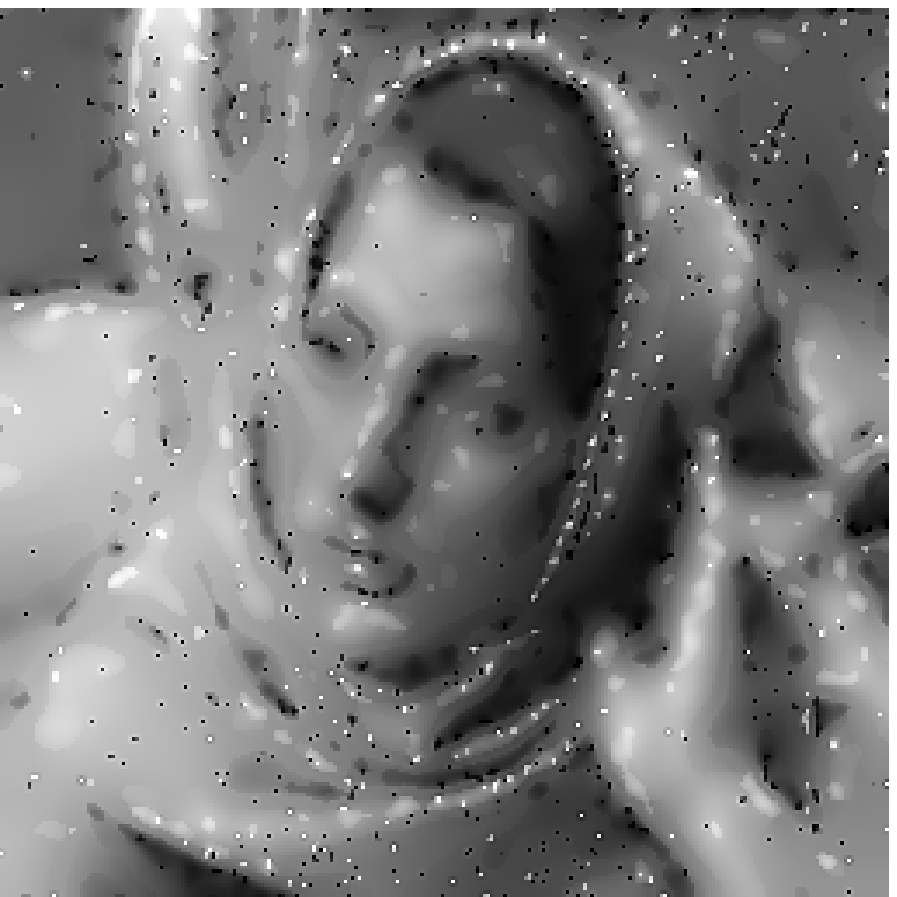} &  &  \\
(d)  EUI+BM3D & (e)  MS-VST + $7/9$ & (f)  MS-VST + B3 &  &
\end{tabular}
\end{center}
}
\caption{{\protect\small Poisson denoising of the Barbara image (image size:
$256 \times 256$). (a) intensity image (intensity $\in [0.93, 15.73])$; (b)
Poisson noisy image; (c) Optimal Weights Filter ($M=15\times15$, $%
m=21\times21$ and $d=0$, $NMISE = 0.1061$); (d) Exact unbiased inverse +
BM3D ($NMISE=0.0863$) (e) MS-VST + $7/9$ biorthogonal wavelet ($J = 4$, $FPR
= 0.001$, $N_{max}= 5$ iterations, $NMISE = 0.2391$); (f) MS-VST + B3
isotropic wavelet ($J = 5$, $FPR = 0.001$, $N_{max}= 5$ iterations, $NMISE =
0.3777$).}}
\label{Fig Barbara}
\end{figure}
\end{center}

\begin{center}
\begin{figure}[tbp]
\renewcommand{\arraystretch}{0.5} \addtolength{\tabcolsep}{-5pt} \vskip3mm {%
\fontsize{8pt}{\baselineskip}\selectfont
\begin{center}
\begin{tabular}{ccccc}
\includegraphics[width=0.25\linewidth]{Cells_original.eps} & %
\includegraphics[width=0.25\linewidth]{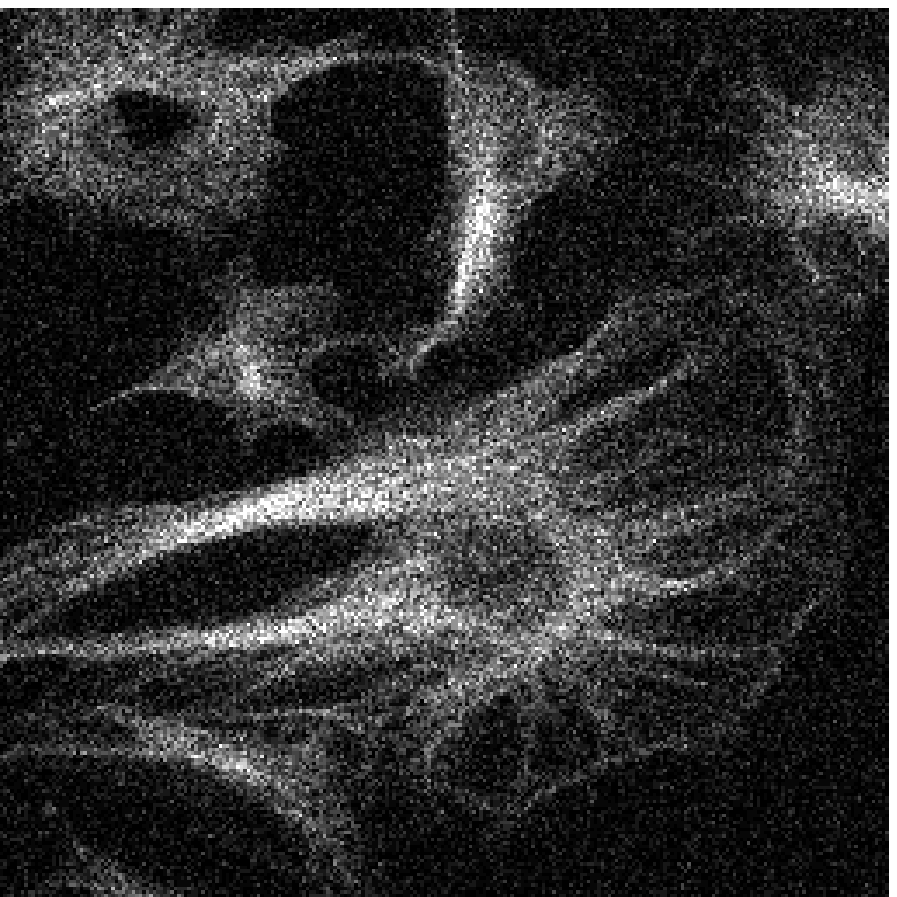} & %
\includegraphics[width=0.25\linewidth]{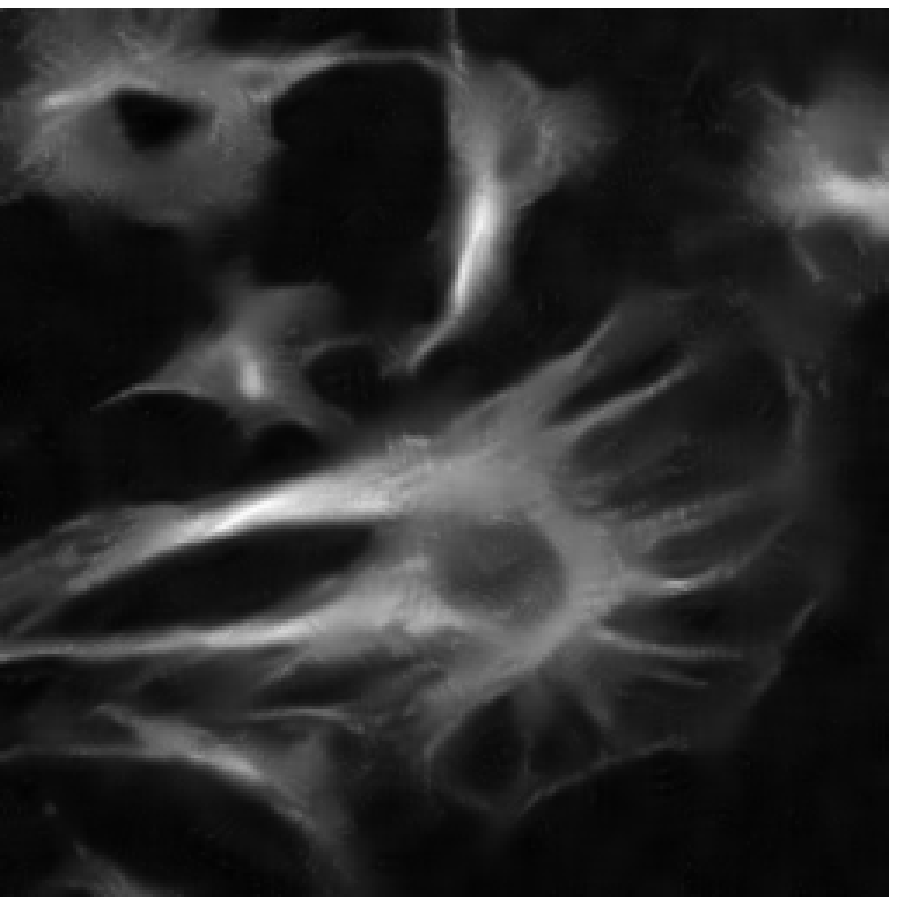} &  &  \\
(a) Original image & (b) Noisy image & (c)  OWF &  &  \\
\includegraphics[width=0.25\linewidth]{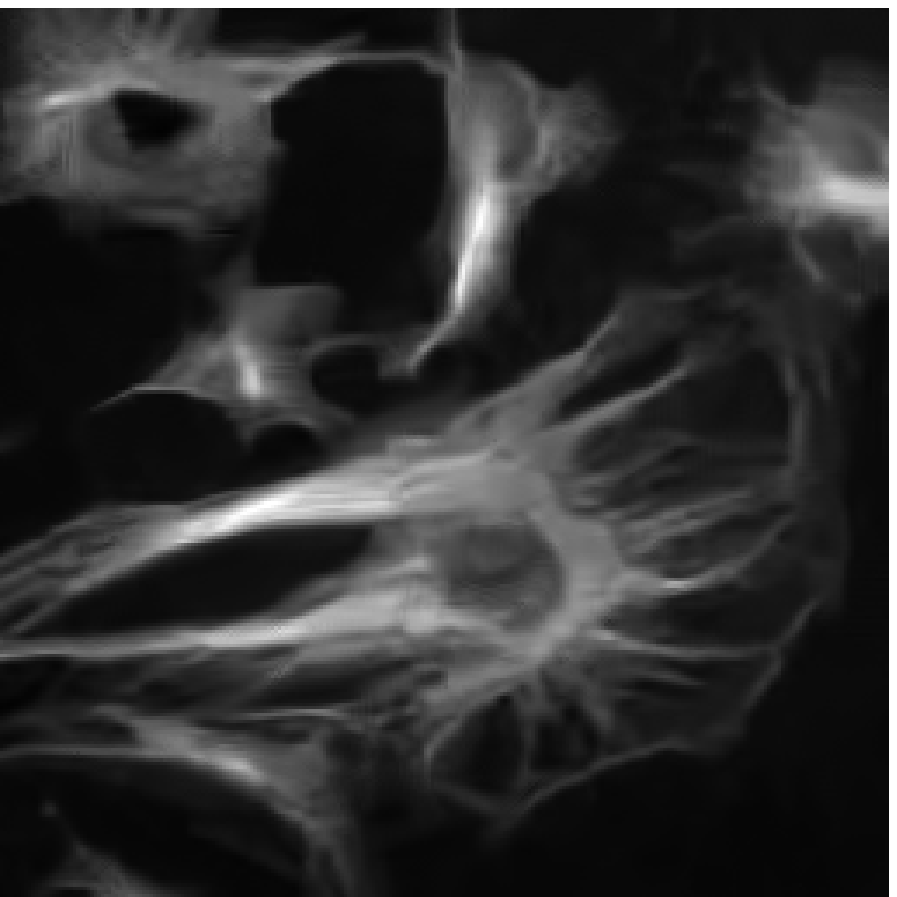} & %
\includegraphics[width=0.25\linewidth]{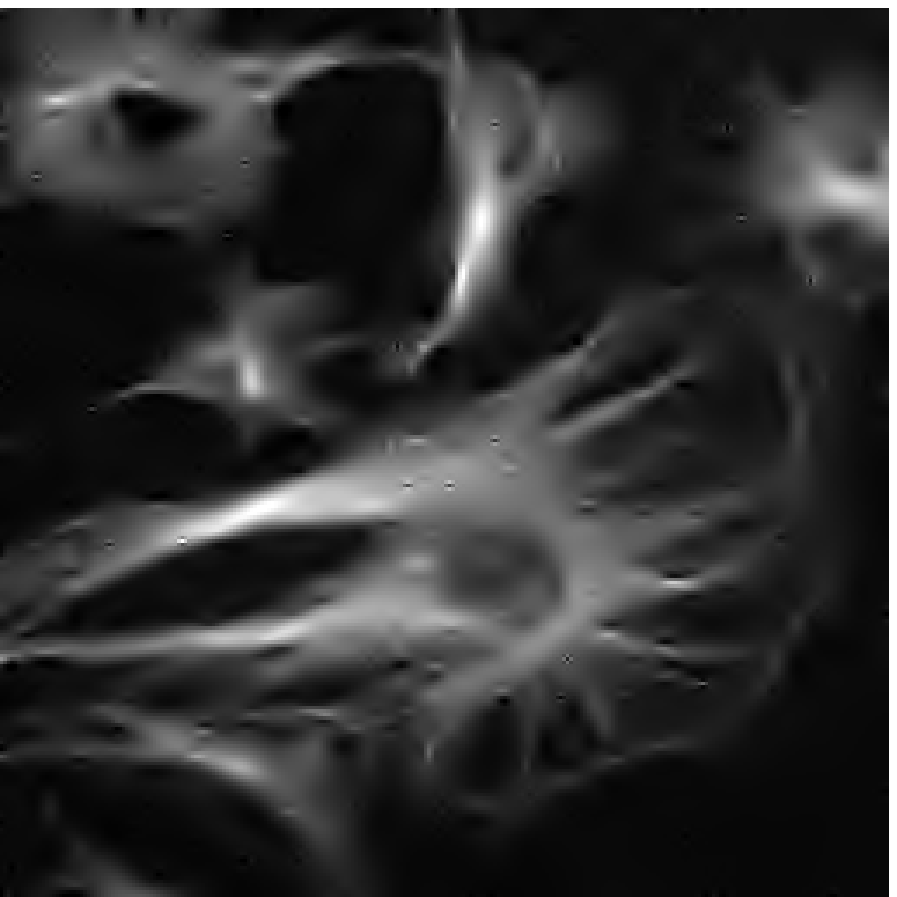} & %
\includegraphics[width=0.25\linewidth]{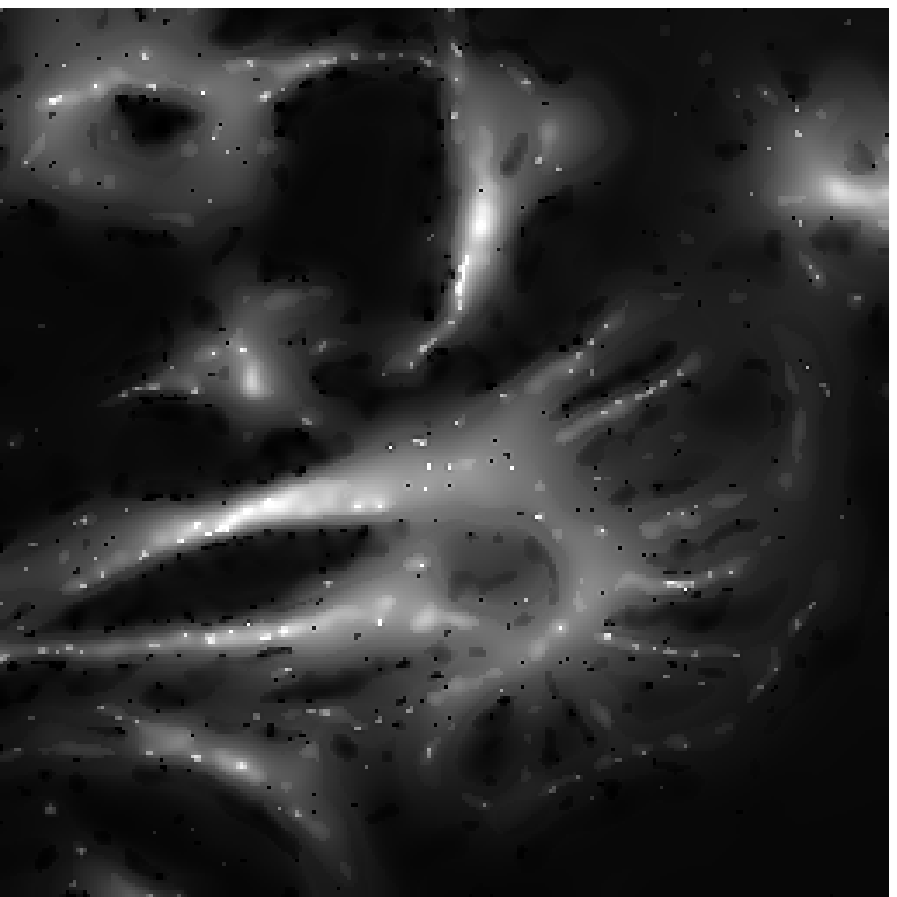} &  &  \\
(d)  EUI+BM3D & (e)  MS-VST + $7/9$ & (f)  MS-VST + B3 &  &
\end{tabular}
\end{center}
}
\caption{{\protect\small Poisson denoising of fluorescent tubules (image
size: $256 \times 256$). (a) intensity image (intensity $\in [0.53, 16.93])$%
; (b) Poisson noisy image; (c) Optimal Weights Filter ($M=11\times11$, $%
m=17\times17$, $d=1$ and $H=0.6$, $NMISE = 0.0794$); (d) Exact unbiased
inverse + BM3D ($NMISE=0.0643$) (e) MS-VST + $7/9$ biorthogonal wavelet ($J
= 5$, $FPR = 0.0001$,$N_{max}= 5$ iterations, $NMISE = 0.0909$); (f) MS-VST
+ B3 isotropic wavelet ($J = 5$, $FPR = 0.001$, $N_{max}= 10$ iterations, $%
NMISE = 0.1487$).}}
\label{Fig Cells}
\end{figure}
\end{center}



\end{document}